\documentclass[times,final]{elsarticle}
\pdfoutput=1

\usepackage{subcaption}

\usepackage{amsmath}
\usepackage{amssymb}

\usepackage{tikz}
\usepackage{quiver}
\usepackage{adjustbox} 

\usepackage{siunitx}

\usepackage{physics}

\usepackage{mathtools}

\usepackage{stmaryrd}

\newcommand{\bR}{\mathbb{R}}

\newcommand{\bT}{\mathbb{T}}

\newcommand{\cC}{\mathcal{C}}
\newcommand{\cD}{\mathcal{D}}
\newcommand{\cE}{\mathcal{E}}

\newcommand{\cO}{\mathcal{O}}
\newcommand{\cP}{\mathcal{P}}

\newcommand{\cR}{\mathcal{R}}
\newcommand{\cS}{\mathcal{S}}

\newcommand{\cX}{\mathcal{X}}

\newcommand{\NamespaceSet}{\mathsf{Namespace}}
\newcommand{\SpaceSet}{\mathbf{S}}
\newcommand{\QuantitySet}{\mathbf{Q}}
\newcommand{\List}{\mathsf{List}}
\newcommand{\Dtry}{\mathsf{Dtry}}
\newcommand{\InterfaceSet}{\mathsf{Interface}}
\newcommand{\SystemSet}{\mathsf{System}}

\DeclareMathOperator{\dom}{dom}

\newcommand\mlnode[1]{\mbox{\begin{tabular}{@{}c@{}}#1\end{tabular}}}

\newcommand{\semantic}[1]{\left\llbracket {#1} \right\rrbracket}

\usepackage[unicode=true]{hyperref}

\hypersetup{%
	plainpages=false,
	bookmarksopen=true,
	bookmarksnumbered=true,
	breaklinks=true,
}

\usepackage{amsthm}

\usepackage[capitalize,nameinlink,noabbrev]{cleveref}

\newtheoremstyle{break} 
  {}          
  {}          
  {}          
  {}          
  {\bfseries} 
  {.}         
  {\newline}  
  {}          
\theoremstyle{break}

\newtheorem{theoremx}{Theorem}[section]

\newtheorem{propositionx}[theoremx]{Proposition}
\newenvironment{proposition}
  {\pushQED{\qed}\propositionx}
  {\popQED\endpropositionx}

\newtheorem{definitionx}[theoremx]{Definition}
\newenvironment{definition}
  {\pushQED{\qed}\definitionx}
  {\popQED\enddefinitionx}

\newtheorem{examplex}[theoremx]{Example}
\newenvironment{example}
  {\pushQED{\qed}\examplex}
  {\popQED\endexamplex}

\newtheorem*{remarkxx}{Remark}
\newenvironment{remark*}
  {\pushQED{\qed}\remarkxx}
  {\popQED\endremarkxx}

\usepackage{booktabs} 

\journal{arXiv}

\newcommand{\titel}{Exergetic Port-Hamiltonian Systems Modeling Language}

\hypersetup{%
	pdfauthor={%
		Markus Lohmayer,
    Owen Lynch
		and Sigrid Leyendecker
		},
  pdftitle=\titel{}
}

\begin{document}

\begin{frontmatter}
	\title{\titel}

	\author[fau]{Markus Lohmayer\corref{cor1}}
	\ead{markus.lohmayer@fau.de}
  \cortext[cor1]{Corresponding author}

	\author[topos]{Owen Lynch}

	\author[fau]{Sigrid Leyendecker}

	\address[fau]{%
		Institute of Applied Dynamics\\
  	Friedrich-Alexander Universität Erlangen-Nürnberg,
		Erlangen, Germany
	}

	\address[topos]{%
		Topos Institute\\
    Berkeley, CA, USA
	}

	\begin{abstract}
		Mathematical modeling of real-world physical systems
		requires the consistent combination of
		a multitude of
		physical laws
		and
		phenomenological models.
		This challenging task can be greatly simplified by
		hierarchically decomposing systems
		into ultimately simple components.
		Moreover, the use of diagrams
		for expressing the decomposition
		helps make the process more intuitive
		and facilitates communication, even with non-experts.
		As an important requirement,
		models have to respect fundamental physical laws
		such as the first and the second law of thermodynamics.
		While some existing modeling frameworks
		make such guarantees
		based on structural properties of their models,
		they lack a formal graphical syntax.
		We present
		a compositional and thermodynamically consistent modeling language
		with a graphical syntax.
		In terms of its semantics,
		we essentially endow port-Hamiltonian systems
		with additional structural properties
		and a fixed physical interpretation,
		ensuring thermodynamic consistency
		in a manner closely related to the metriplectic or GENERIC formalism.
		While port-Hamiltonian systems are inspired by
		graphical modeling with bond graphs,
		neither the link between the two,
		nor bond graphs themselves,
		can be easily formalized.
		In contrast,
		our syntax is based on a refinement of
		the well-studied operad of undirected wiring diagrams.
		By combining
		a compositional, graphical syntax
		with
		an energy-based, thermodynamic approach,
		the presented modeling language
		simplifies
		the understanding, reuse, and modification
		of complex physical models.
	\end{abstract}

\end{frontmatter}

\section{Introduction}%

In this introduction
we mainly want to answer the following questions:
\begin{enumerate}
  \item{What are Exergetic Port-Hamiltonian Systems (EPHS)?}
  \item{How do they enable compositional modeling?}
  \item{Why is this practically relevant?}
\end{enumerate}
In addition,
we briefly summarize key ideas of nonequilibrium thermodynamics,
and
we discuss related work that inspired our developments.
The introduction finishes with an outline of the main part of the paper.

\subsection{What?}%

Exergetic Port-Hamiltonian Systems (EPHS) provide
a compositional and thermodynamically consistent
language for expressing mathematical models
of multiphysical systems at macroscopic scales.
Such systems may combine
\begin{enumerate}[(a)]
  \item{%
    classical mechanics\\
    (e.g.~spring-mass systems, multibody dynamics, fluid mechanics),
  }
  \item{%
    classical electromagnetism\\
    (e.g.~LC circuits, electromagnetic wave propagation)
  }
  \item{%
    and irreversible processes with local thermodynamic equilibrium\\
    (e.g.~friction/viscosity, electrical resistance/conduction, thermal conduction).
  }
\end{enumerate}

The EPHS language is termed compositional
because it enables
a fully modular and hierarchical approach to model specification
through its simple, graphical syntax,
as elaborated in the next subsection.

Flattening any hierarchical nesting reveals that
every EPHS model is fundamentally defined by
a power-preserving interconnection of primitive subsystems.
There are three main classes of primitive EPHS:
\begin{enumerate}
  \item{%
    Systems that represent energy storage
    are characterized by
    a function which yields the stored energy
    for each state of the primitive system.
  }
  \item{%
    Systems that represent reversible energy exchange
    are characterized by
    a so-called Dirac structure~\cite{1990Courant},
    as known e.g.~from (constrained) mechanical systems.
  }
  \item{%
    Systems that represent irreversible energy exchange
    are characterized by
    another type of structure that we call Onsager structure.
  }
\end{enumerate}
Since both classical mechanics and electromagnetism
can be expressed using Dirac structures,
item 2.~corresponds to (a) and (b) above,
while item 3.~corresponds to (c).
Structural properties of the primitive systems
and their interconnection
guarantee that all models respect
the first and the second law of thermodynamics,
as well as Onsager's reciprocal relations.

Inspired by the contribution in~\cite{2018BadlyanMaschkeBeattieMehrmann},
previous work~\cite{2021LohmayerKotyczkaLeyendecker}
has explored EPHS as
a systematic approach to thermodynamic modeling
based on port-Hamiltonian systems
(see, e.g.,~\cite{2014SchaftJeltsema}).
By interpreting the `Hamiltonian'
as an exergy storage function,
a stricter version of port-Hamiltonian structure can be formulated,
ensuring
a thermodynamically consistent
reversible-irreversible splitting of the dynamics,
as known from the metriplectic or GENERIC formalism
(see, e.g.,~\cite{1984Morrison,2018PavelkaKlikaGrmela}).
The exergy storage function represents
the theoretically `useful' amount of energy,
i.e.,~the maximum amount of work
that can be extracted
from a thermodynamic system
using an ideal device,
such as a Carnot engine.
This ideal device has access
not only to all forms of energy stored in the system
but also to a (reference) environment,
primarily characterized by
a heat bath with constant temperature.

The present work formalizes EPHS
as a compositional modeling language
with a simple, graphical syntax
inspired by bond graphs
(see, e.g.,~\cite{1961Paynter,2010Borutzky}).
This approach enables the construction of complex models
by composing them from simpler subsystems,
which can themselves be constructed from
more basic subsystems, and so on.
We draw on and extend ideas form (applied) category theory
(see, e.g.,~\cite{2019FongSpivak,1998Maclane})
to propose a simple mathematical framework
suitable for computer implementation.
In principle, the graphical syntax
and the associated compositional framework
could also be adopted without
the parts that ensure thermodynamic consistency.
While this would relax certain conditions
imposed on the primitive systems,
it would also weaken the correspondence between
the remaining structure and its physical interpretation.

\subsection{How?}%
\label{ssec:syntax}

\begin{figure}[ht]
  \centering
  \includegraphics[width=15em]{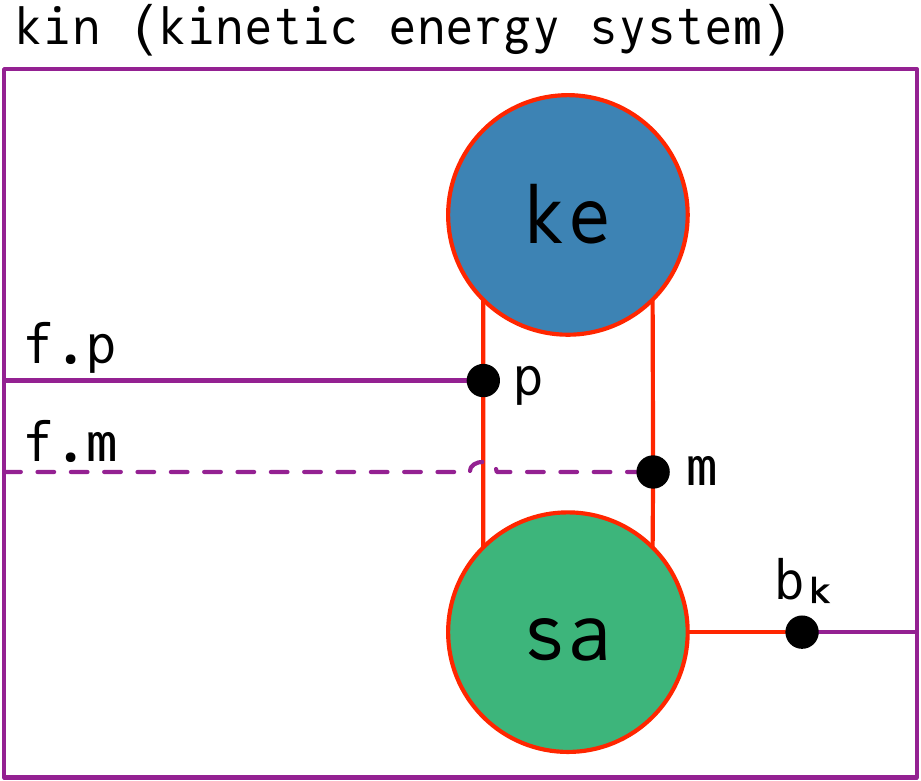}%
  \caption{%
    Graphical display of
    an interconnection pattern
    with different syntactic elements
    shown in different colors.
    The outer box
    shown in {\color[HTML]{942192}purple}
    represents the interface of the resulting composite system.
    The three protruding lines
    represent its ports
    called $\mathtt{f.p}$, $\mathtt{f.m}$ and $\mathtt{b_k}$.
    The two inner boxes
    shown in {\color[HTML]{FF2600}red}
    represent (the interface of)
    the subsystems called
    $\mathtt{ke}$ (storage of kinetic energy)
    and
    $\mathtt{sa}$ (self-advection of kinetic energy).
    Both have ports called
    $\mathtt{p}$ (momentum) and $\mathtt{m}$ (mass).
    Box $\mathtt{sa}$ has another port called
    $\mathtt{b_k}$
    (advection of kinetic energy across the boundary of the spatial domain).
    The common prefix $\mathtt{f}$ (fluid)
    of the outer ports
    $\mathtt{f.p}$ (fluid momentum)
    and
    $\mathtt{f.m}$ (fluid mass)
    combines them into a multiport.
    The black dots represent junctions
    where power is exchanged among the connected ports.
    The dashed line indicates that
    only information about the state (fluid mass)
    is exchanged via
    the outer port $\mathtt{f.m}$,
    with no energy transfer occurring.
  }%
  \label{fig:kin_colored}
\end{figure}

Expressions in the graphical syntax
are called interconnection patterns,
since such a pattern specifies
how a composite system is formed
through a power-preserving interconnection of subsystems.
\Cref{fig:kin_colored} shows
an example from fluid dynamics,
see~\cite{2024LohmayerKrausLeyendecker}.

The syntax deals with systems only in terms of
their interfaces, i.e.~their exposed ports.
Each port has two attributes:
The first attribute gives
the physical quantity that can be exchanged via the port
together with the corresponding state space.
To rule out ill-defined connections,
only ports with the same quantity can be connected.
The second attribute is of Boolean type
and indicates whether
power can be exchanged over the port
or merely information about system state,
as indicated graphically by a solid or a dashed line, respectively.

We think of an interconnection pattern as a function.
For each inner box,
it takes as input a system with matching interface
and it outputs a system whose interface matches the outer box.
For instance,
the pattern shown in~\cref{fig:kin_colored}
represents a function
that takes two subsystems,
whose interfaces are determined by
the ports of the boxes $\mathtt{ke}$ and $\mathtt{sa}$,
and it combines them into a composite system,
whose interface is given by the ports of the outer box.
We draw inner boxes filled with a color
to indicate
the nature of the subsystems
that we consider as inputs.
In~\cref{fig:kin_colored},
box $\mathtt{ke}$ has a {\color[HTML]{346F99}blue} filling,
as it represents energy storage
(the fluid's kinetic energy).
Box $\mathtt{sa}$ has a {\color[HTML]{349A69}green} filling,
as it represents reversible energy exchange
(advection of momentum and mass).
Written on top of the outer box,
\texttt{kin} is an identifier
akin to a function name
and the text in parentheses is just a short description.
The multiport $\mathtt{f}$
facilitates integrating the system
into a more complex model
without drawing parallel connections.

\begin{figure}[ht]
  \centering
  \subcaptionbox{Kinetic energy system}{
    \includegraphics[width=.3\textwidth]{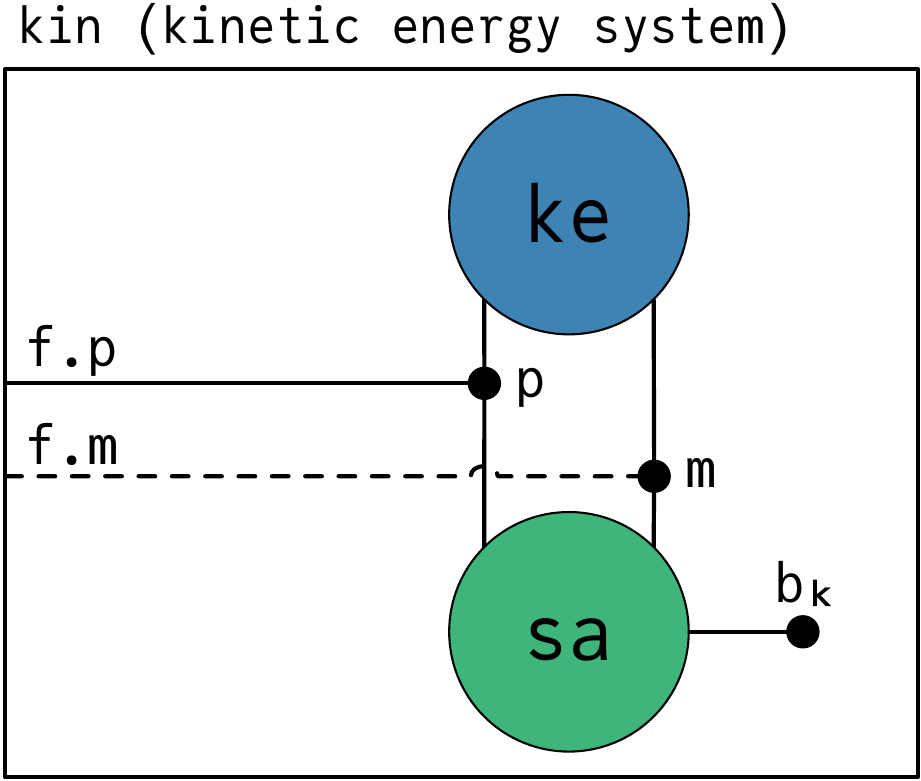}
  }
  \subcaptionbox{Fluid system}{
    \includegraphics[width=.3\textwidth]{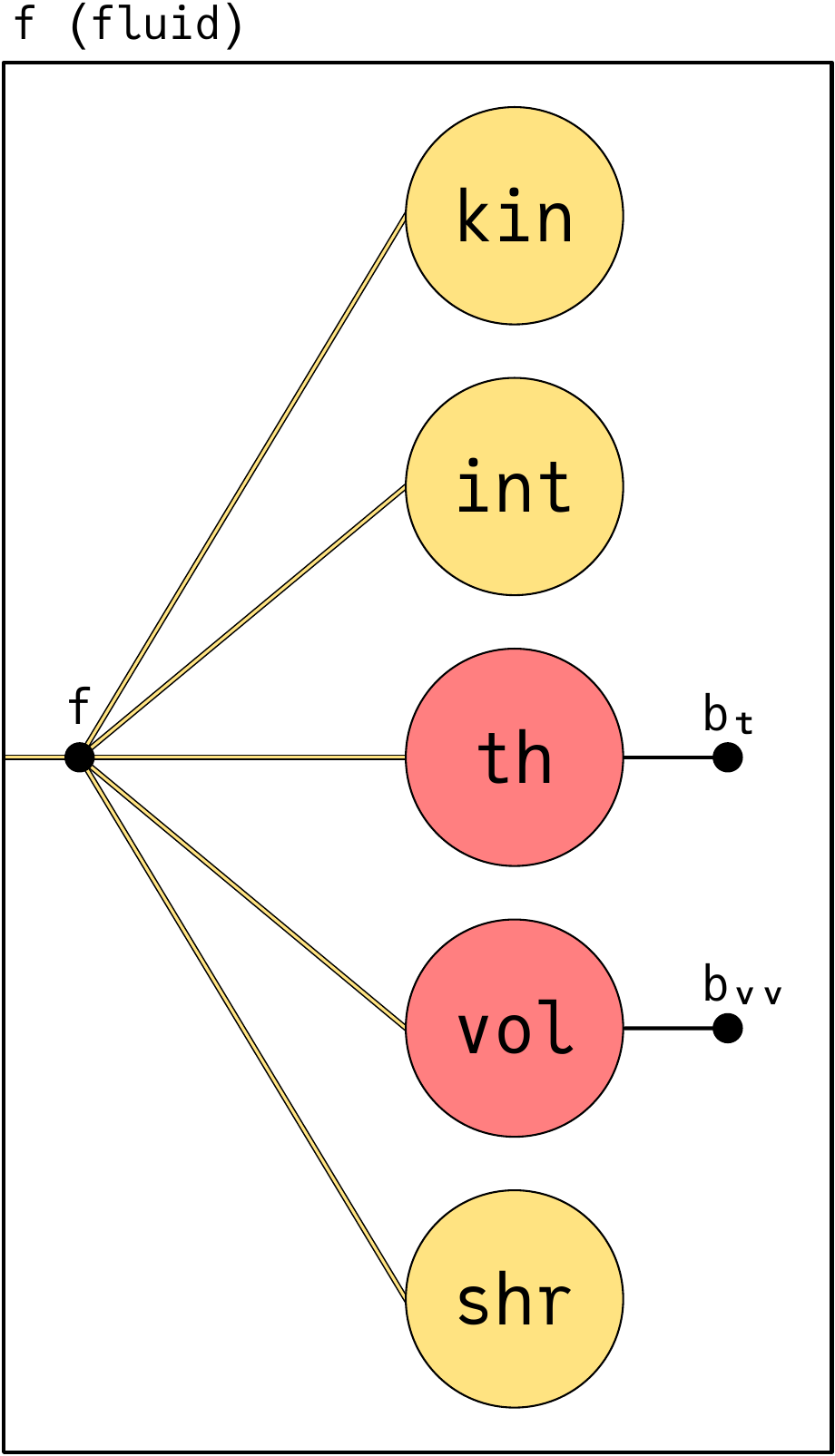}
  }
  \subcaptionbox{EMHD system}{
    \includegraphics[width=.3\textwidth]{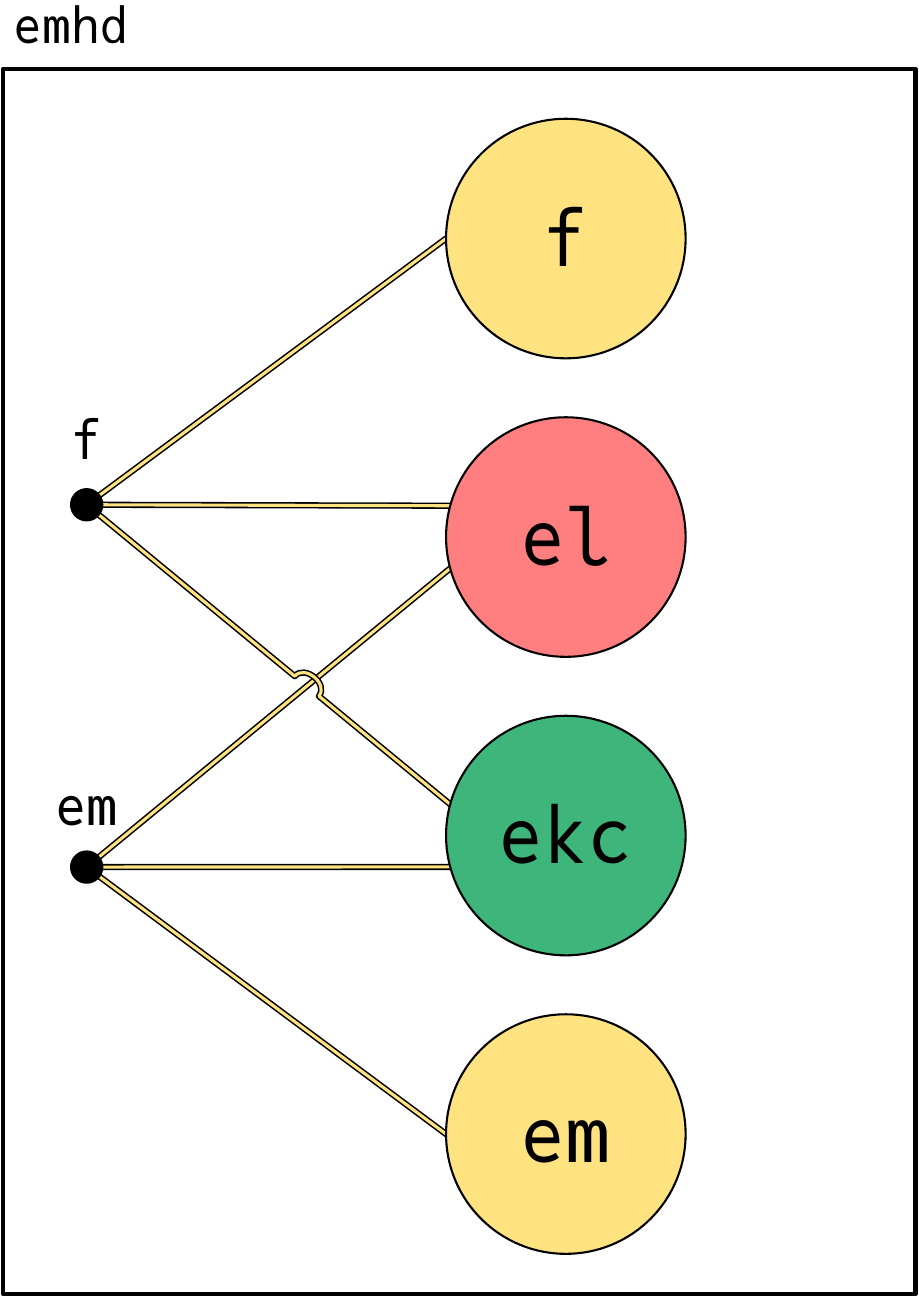}
  }
  \caption{%
    Hierarchical specification of
    an electro-magneto hydrodynamics (EMHD) model~%
		\cite{2024LohmayerKrausLeyendecker}:
		The fluid model \texttt{f} shown in (b) includes
		the kinetic energy system \texttt{kin} depicted in (a),
		as well as four additional models not detailed here:
		an internal energy system \texttt{int},
		a thermal conduction model \texttt{th},
		and volume and shear viscosity models \texttt{vol} and \texttt{shr}.
    The electro-magneto hydrodynamics model \texttt{emhd}
    shown in (c) includes
		the fluid model \texttt{f} depicted in (b),
		as well as three additional models not detailed here:
    an electromagnetic system \texttt{em},
    an electric conduction model \texttt{el},
    and a model for the electro-kinetic coupling \texttt{ekc}.
  }%
  \label{fig:emhd}
\end{figure}
\Cref{fig:emhd} shows
the hierarchical specification of a plasma model.
For simplicity,
a spatial domain with isolated boundary
is assumed
such that
only the in-domain interconnection of the different models
needs to be dealt with.
A {\color[HTML]{D96C6D}red} filling indicates that
the respective input is
a primitive system modeling an irreversible process.
A {\color[HTML]{D9C16D}yellow} filling signifies that
the input is a nested composite system,
specified in terms of its own pattern.
Recursively substituting
into all yellow boxes
the respective patterns
would yield an equivalent model specification
without the hierarchical decomposition.
The composition of patterns
through substitution is
uniquely defined whenever interfaces match
(possibly up to a given renaming of the ports).
No additional data or handling of `edge cases' is required.
One can hence easily add levels of abstraction
by refactoring subsystems
into reusable parts.
In the fluid model,
one could for instance factor out
the kinetic energy system
filling box $\mathtt{kin}$
and
the internal energy system
filling box $\mathtt{int}$
into a reusable model of an ideal compressible fluid,
see~\cite{2024LohmayerKrausLeyendecker}.

\subsection{Why?}%

EPHS provide
a formal and intuitive modeling language
suited for humans and software tools alike.
As shown above,
increasingly complex models
are assembled from simpler
and ultimately primitive parts.
This fosters the reusability of models
and encourages the conceptual separation
of different levels of detail.
Beyond the modeling process,
this can be beneficial when
the language is used in
interdisciplinary communication
and education.
Some modeling errors can be completely avoided due to
the built-in thermodynamic consistency
and
the typed interfaces
(given by ports with fixed physical quantities).
Since interconnection patterns are essentially
energy/exergy flow diagrams,
EPHS models are particularly useful for
thermodynamic analysis and optimization.
The compositional nature of EPHS
could enable
a scalable approach
to model transformations, such as discretizations.
It could also enable a practical integration of
first-principles-based thermodynamic modeling
and scientific machine learning.

\subsection{Nonequilibrium thermodynamics in a nutshell}%
\label{ssec:thermodynamics}

Most physical systems
are characterized by
a huge number of degrees of freedom
but
keeping track of them
is usually neither possible nor useful.
Nonequilibrium thermodynamics aims to
provide a framework for coarse-grained modeling,
meaning that
fast dynamics at small scales are not resolved,
and
a quantity called entropy is introduced to account for
the missing information,
see, e.g.,~\cite{2018PavelkaKlikaGrmela}.
Coarse-graining is possible because,
despite fluctuations at the microscopic scale,
one can eventually observe non-fluctuating properties
at a more macroscopic scale;
at least this is the case for isolated systems.
For instance,
temperature is a property
that only appears at more macroscopic scales.
Thermodynamic properties are classified
as extensive
when they are proportional to the size of the system
(e.g.~entropy, volume, mass)
or as intensive
when they do not depend on the size of the system
(e.g.~temperature, pressure, chemical potential),
see, e.g.,~\cite{1985Callen}

A system is said to be
in local thermodynamic equilibrium if
for every material point
there exists a neighborhood
large enough
to observe thermodynamic properties,
yet small enough
for the properties to be essentially non-varying inside the neighborhood,
see, e.g.,~\cite{2005Oettinger}.
The size of these neighborhoods
corresponds to a macroscopic scale
at which the dynamics of the thermodynamic properties can be observed
and also modeled.

The first law of thermodynamics
restates a fact known from mechanics,
namely that energy is a conserved quantity.
However,
instead of considering
directly the mechanical and electromagnetic energy
associated to microscopic degrees of freedom,
a thermodynamic model with local equilibrium
resorts to a phenomenological description.
This so-called internal energy
is considered to depend on
a number of extensive properties,
including entropy,
and possibly further mechanical or electromagnetic state variables
which are observable on the considered macroscopic scale.
These quantities then serve as
the state variables of a coarse-grained description.
Moreover,
temperature is defined as the derivative of
the internal energy with respect to entropy.
Similarly,
pressure is given by the derivative with respect to volume
and the chemical potential is the derivative with respect to mass,
see, e.g.,~\cite{1985Callen}.

Models in classical mechanics and electromagnetism
exhibit reversible dynamics,
meaning that one could in principle arrive again at
the initial condition by formally reversing the direction of time
(although for many systems this would be practically impossible due to chaos).
Incomplete knowledge about the microscopic state
immediately leads to an irreversible dynamics
and
an irreversible dynamics in turn leads to a growth of entropy,
as information
(e.g.~about the initial state of the system)
is lost as the system evolves.
This is reflected by
the second law of thermodynamics
which states that
the entropy of an isolated system never decreases
and that the system evolves towards a state of maximum entropy,
see, e.g.,~\cite{2005Oettinger}.
In other words,
the uncertainty about the microscopic state
can only grow, not decrease.

For models with local equilibrium,
the irreversible dynamics
can be expressed in terms of
so-called thermodynamic fluxes
that determine instantaneous changes
of coarse-grained state variables.
The fluxes are expressed as
a function of
so-called thermodynamic forces,
which are given by local differences in intensive quantities,
see, e.g.,~\cite{2005Oettinger,1984GrootMazur}.
For instance,
temperature differences cause heat flux,
according to Fourier's law,
and
differences in chemical potential cause a diffusive mass flux,
according to Fick's law.
Further,
one may observe cross effects such as
a diffusive mass flux caused by temperature differences (Soret effect)
and likewise
a heat flux caused by differences in chemical potential (Dufour effect).
Onsager's reciprocal relations state that
the relation between fluxes and forces
must possess a certain symmetry~\cite{1931Onsager}.
Specifically,
the ratio of
the mass flux caused by the Soret effect
and
the respective temperature difference
is equal to
the ratio of
the heat flux caused by the Dufour effect
and
the respective difference in chemical potential.

\subsection{Related work}%

\Cref{fig:related_work} gives
a high-level view of this subsection.
\begin{figure}[ht]
  \centering
  \includegraphics[width=0.78\textwidth]{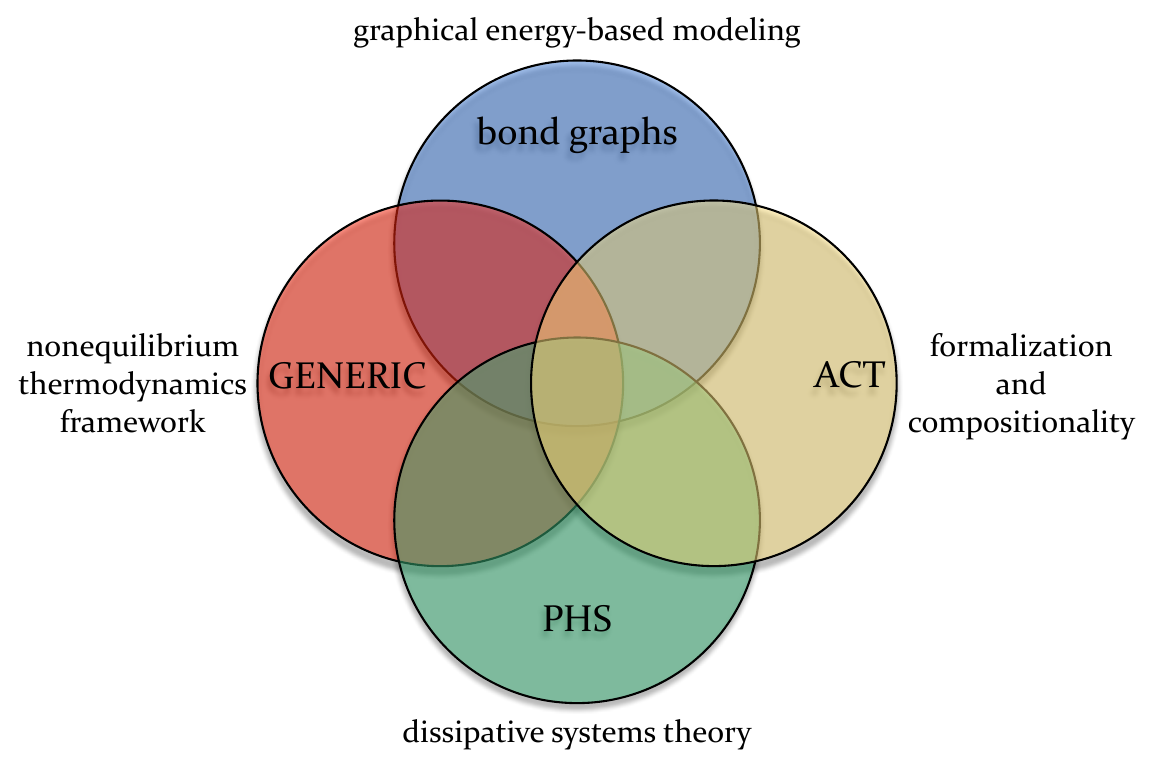}
  \caption{%
    Informal Venn diagram illustrating that
    EPHS integrates ideas from
    graphical and energy-based modeling
    of physical systems
		with bond graphs
		(see, e.g.,~\cite{1961Paynter,2010Borutzky}),
    the metriplectic/GENERIC framework for
		nonequilibrium thermodynamics
		(see, e.g.,~\cite{1984Morrison,2005Oettinger}),
    port-Hamiltonian theory
    for open, dissipative systems (PHS)
		(see, e.g.,~\cite{2014SchaftJeltsema}),
    and applied category theory (ACT) research on
    the formalization of graphical languages
    and compositional dynamical systems
		(see, e.g.,~\cite{2019FongSpivak,2013Spivak,2015BaezErberle,2020Libkind,2023Myers}).
  }%
  \label{fig:related_work}
\end{figure}

\subsubsection{Bond graphs}

Bond graphs provide
a graphical notation
for representing multi-domain models,
particularly those
involving mechanical and electrical systems,
see, e.g.,~\cite{1961Paynter,2010Borutzky}.
The main idea is that elements interact
via so-called bonds that represent links for energy exchange.
Each bond is associated with
two dynamic variables,
called flow and effort,
whose product
gives the exchanged power.
Bonds can be connected through two types of junctions:
$0$-junctions,
where the flow variables of all connected ports balance
while all effort variables are equal,
and $1$-junctions,
where the effort variables balance
while all flow variables are equal.
In both cases,
the net power at a junction is zero.
The elements are commonly generalized
capacitors, inductors, transformers, resistors, etc.
The junctions are then seen as
generalized Kirchhoff's circuit laws.

\subsubsection{Port-Hamiltonian systems}

Port-Hamiltonian Systems (PHS)
were originally inspired by
relating
the dynamical systems obtained from
\textit{generalized bond graphs}\footnotemark
to Hamiltonian mechanics~\cite{1991Maschke}.
\footnotetext{
  While bond graphs commonly use
  generalized capacitors and inductors as storage elements,
  so-called \emph{generalized bond graphs} use
	only generalized capacitors
  together with gyrator elements~\cite{1982Breedveld}.
}%
While many authors make use of
PHS in input-state-output form,
PHS generally constitute an acausal systems theory,
also inspired by
Willems' behavioral approach~\cite{1974Willems,2007Willems}.
The dynamics of networks composed of
energy storing elements
and reversible coupling elements
(generalized transformers and gyrators)
can,
from a geometric perspective,
be formalized
as Hamiltonian systems
whose Poisson structure
reflects the network topology~%
\cite{1992MaschkeSchaftBreedveld}.
The generalization to systems with
algebraic constraints
and/or external port variables
is achieved
using Dirac structures~%
\cite{1995SchaftMaschke}.
These offer a relational approach that unifies
Poisson and presymplectic structures
(see, e.g.,~\cite{1990Courant,2013Bursztyn}).
By modeling resistive elements
as relations imposed on external port variables,
Hamiltonian dynamics can be combined with
(dissipative) gradient dynamics.
A key feature of PHS is their composability:
interconnecting simpler systems
yields a more complex one.
This is achieved by composing the Dirac structures
of the individual systems with an
`interconnecting Dirac structure'
(see, e.g.,~\cite{2014SchaftJeltsema}).
Another key feature is that
every PHS satisfies a power-balance equation,
stating that
at every instant
the sum of the stored energy and the dissipated energy
equals the supplied energy.
As the dissipated energy cannot be negative,
PHS are inherently passive.
In this context, the term `energy'
stems from the `Hamiltonian' storage function.
For EPHS, this function is interpreted as the stored exergy,
making the power-balance equation
a statement about rates of exergy transfer.
This interpretation fully reconciles
the inherent passivity of PHS
with the first and second law of thermodynamics,
as exergy is indeed dissipated (or destroyed) by irreversible dynamics.

\subsubsection{Metriplectic and GENERIC framework}

Metriplectic systems
as well as
the \emph{general equation for reversible-irreversible coupling} (GENERIC)
refer to a framework
for combining reversible and irreversible dynamics
in a thermodynamically consistent manner,
see, e.g.~\cite{1984Morrison,1984Grmela1,1984Grmela2,1984Kaufman,1997GrmelaOettinger}.
The reversible part is modeled as
a Hamiltonian dynamics,
which inherently conserves energy,
while
the irreversible part is modeled as
a (generalized) gradient flow,
which inherently drives the system toward maximum entropy.
Two additional conditions ensure that
the reversible dynamics conserves entropy,
and the irreversible dynamics conserves energy.
As a result, the combined dynamics respects
both the first law and the second law
(see, e.g.,~\cite{2005Oettinger,2018PavelkaKlikaGrmela}).
While many well-known physical models
can be identified as
metriplectic systems or instances of the GENERIC,
the framework is particularly useful
as a guide for developing new thermodynamic models.

\subsubsection{Applied category theory}

Category theory
was first applied in
algebraic geometry and topology
and now is used to
organize many branches of mathematics
(see, e.g.,~\cite{1998Maclane}).
In recent decades,
category theory has grown beyond pure math
to organize scientific and engineering disciplines,
giving birth to
applied category theory
(see, e.g.~\cite{2019FongSpivak}).
The following ideas from
categorical systems theory
(see, e.g.,~\cite{2023Myers})
have inspired our work.
\begin{enumerate}
  \item{%
    The shift of focus from
    individual systems to
    how systems relate to each other.
  }
  \item{%
    An isolated system is
    a system with empty interface.
  }
  \item{%
    The description of systems involves
    \emph{syntax} and \emph{semantics}.
    Syntax consists of
    combinatorial data
    that can be manipulated computationally,
    while semantics are given by geometric objects
    that exist only platonically.
    %
    %
    An explicit understanding of the translation
    from syntax to semantics
    serves as a guide for computer implementation.
  }
\end{enumerate}

A prominent line of research
formalizes graphical languages such as
block diagrams~\cite{2015BaezErberle},
Petri nets~\cite{2017BaezPollard},
junctions of bond graphs~\cite{2017Coya}, and
passive linear circuits~\cite{2018BaezFong}.
In each case,
the graphical syntax is defined as
some sort of string diagrams,
using a generators and relations approach.

Another line of research uses the theory of operads to
define a graphical syntax
that supports hierarchical (de-)composition
(see, e.g.,~\cite{%
2020BaezFoleyMoellerPollard,%
2021FoleyBreinerSubrahmanianDusel%
}).
Operads organize formal operations
with finitely many inputs (subsystems)
and one output (composite system)
(see, e.g.,~\cite{2004Leinster,2016Yau}).
As proposed in~\cite{2022LohmayerLeyendecker1},
the EPHS syntax,
is based on
the operad of undirected wiring diagrams,
which initially appeared in~\cite{2013Spivak},
together with a relational semantics
used to express database queries,
see also~\cite{2014Spivak,2018Yau}.
Directed wiring diagrams have been used to formalize
nested systems of
Moore machines and ODEs with inputs and outputs~%
\cite{%
2015VagnerSpivakLerman,%
2019SchultzSpivakVasilakopoulou%
}.
A framework for
discrete and continuous dynamical systems
with directed and undirected notions of composition
is presented in~\cite{2021LibkindBaasPattersonFairbanks}.
The composition of
lossless PHS
is studied
from a categorical perspective
in~\cite{2022Lynch}.

\subsection{Outline}%

\Cref{sec:interfaces}
defines system interfaces.
Based on this,
\cref{sec:interconnection-patterns}
defines interconnection patterns.
\Cref{sec:interconnection-semantics}
shows how an interconnection pattern
implies a relation between the port variables
of its interfaces.
\Cref{sec:exergy}
discusses how
exergy quantifies
the energy stored in a system
that is available for doing work.
Based on this,
\cref{sec:primitive-systems}
introduces the primitive systems.
\Cref{sec:composite-systems}
defines composite systems
and
how they are formed
using the graphical syntax.
\Cref{sec:thermodynamic-consistency}
shows that EPHS models are thermodynamically consistent.
\Cref{sec:example} presents
an example of
a direct current (DC) motor model,
along with
a comparison to
a classical bond graph.
\Cref{sec:discussion} concludes with
a discussion.

\section{Interfaces}%
\label{sec:interfaces}

While the narrative of discovery,
usually proceeds from systems to interfaces,
logically speaking,
the syntax needs to be introduced before the semantics
and hence we start with interfaces.
Graphically,
the interface of a system is a box
with a set of ports drawn around it.
While we could treat such sets abstractly,
as is often done in applied category theory,
for practical reasons
we instead treat them as sets of names
called \emph{namespaces}.

\subsection{Namespaces}%

\begin{definition}
	A \textbf{string} is a non-empty list of characters
	and
	a \textbf{name} is a list of strings.
\end{definition}

We write names
using a monospace font
and
with strings separated by dots.
For instance,
the name
consisting of the list of strings ``oscillator'', ``spring'', ``q''
is written as $\mathtt{oscillator.spring.q}$.
Further,
we use $[]$ to denote the empty name,
i.e.~the empty list of strings.
Given two names $n$ and $n'$,
we denote their concatenation by $n \ast n'$.
For example,
$\mathtt{a.b} \ast \mathtt{c} = \mathtt{a.b.c}$.

\begin{definition}
	Given two names $p$ and $n$,
	we say that $p$ is a \textbf{prefix} of $n$
	and write $p \leq n$
	if there exists any name $s$ such that $n = p \ast s$.
	In case $s \neq []$,
	we say that $p$ is a \textbf{strict prefix} of $n$
	and write $p < n$.
\end{definition}

\begin{example}
	It holds that
	$\mathtt{a.b} < \mathtt{a.b.c}$
	and
	$\mathtt{a.b} \leq \mathtt{a.b}$.
	Also,
	$[] < \mathtt{a}$.
	It does \emph{not} hold that
	$\mathtt{a.a} \leq \mathtt{a}$.
\end{example}

\begin{definition}
	A set of names $N$ is \textbf{prefix-free}
	if for all $n, n' \in N$,
	$n$ is \emph{not} a strict prefix of $n'$.
	A \textbf{namespace} is a finite, prefix-free set of names.
	Let $\NamespaceSet$ refer to the set of all namespaces.
\end{definition}

\begin{example}
	The set of names
	$\{ \mathtt{a}, \, \mathtt{b} \}$ is a namespace
	and so is
	$\{ \mathtt{a.a}, \, \mathtt{a.b} \}$.
	Since
	$\mathtt{a} < \mathtt{a.b}$,
	the set of names
	$\{ \mathtt{a}, \, \mathtt{a.b} \}$
	is not prefix-free and hence not a namespace.
\end{example}

Sets of names are equivalent to prefix trees,
called \emph{tries} in computer science~\cite{1960Fredkin}.

\begin{example}
	The namespace
	$\{ \mathtt{a.b.b}, \, \mathtt{a.b.c}, \, \mathtt{b.a}, \, \mathtt{b.y}, \, \mathtt{c} \}$
	corresponds to the following trie:
	\begin{equation*}
		\includegraphics{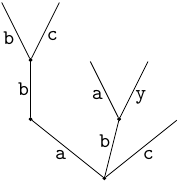}
		\qedhere
	\end{equation*}
\end{example}

Note that
for a namespace $N$,
the leaf nodes of its trie correspond to the names in $N$,
while for a set of names that is not prefix-free,
at least one name corresponds to an internal node.
In the following,
we think of namespaces both as sets of names and as tries.

Given several namespaces,
we could take their union as sets.
However, this runs the risk of \emph{name conflicts},
as it is possible for a given name in the union to have come from either of the original sets.
It is for this reason that we introduce the concept of a \emph{named sum},
which can be thought of as a practical implementation of \emph{disjoint union}.

\begin{definition}
	Given a namespace $N$,
	and a function $f \colon N \to \NamespaceSet$,
	the \textbf{named sum} of $f$ is
	\begin{equation*}
		\sum f
		\: = \:
		\Bigl\{
			\,
			n \ast a
			\, \mid \,
			n \in N, \:
			a \in f(n)
			\,
		\Bigr\}
		\,.
		\qedhere
	\end{equation*}
\end{definition}

We also sometimes write expressions like
$\sum \left[ \mathtt{a} \mapsto I_1, \, \mathtt{b} \mapsto I_2 \right]$,
which refers to the named sum of the anonymous function defined by
\begin{equation*}
	x
	\: \mapsto \:
	\begin{cases}
		I_1 & \text{if } x = \mathtt{a} \\
		I_2 & \text{if } x = \mathtt{b}
		\,.
	\end{cases}
\end{equation*}

\begin{example}
	Given
	$N = \{ \mathtt{p}, \, \mathtt{q} \}$,
	we have
	\begin{equation*}
		\sum \left[ \mathtt{l} \mapsto N, \, \mathtt{r} \mapsto N \right]
		\: = \:
		\bigl\{
		\mathtt{l.p}, \, \mathtt{l.q}, \,
		\mathtt{r.p}, \, \mathtt{r.q}
		\bigr\}
		\,.
		\qedhere
	\end{equation*}
\end{example}

\begin{proposition}
	Given a namespace $N$
	and a function
	$f \colon N \to \NamespaceSet$,
	the named sum
	$\sum f$ is again a namespace, that is to say, it is prefix-free.
\end{proposition}
\begin{proof}
	The trie $\sum f$ is simply the trie $N$ with,
	for each leaf node $n \in N$,
	the trie $f(n)$ grafted onto it.
\end{proof}

As shown next,
we can use namespaces and named sums
to develop a simple and effective formalism
for working with hierarchically-structured data.

\subsection{Directories}%

The interface of a system is given by a namespace,
where each name refers to a port.
However, this is not all,
as each port has two associated attributes.
This motivates the definition of \emph{directories}:

\begin{definition}
	Given a set $A$ (of possible attributes),
	we define
	\begin{equation*}
		\Dtry(A)
		\: = \:
		\Bigl\{
			\,
			(N, \, \tau)
			\, \mid \,
			N \in \NamespaceSet, \,
			\tau \colon N \to A
			\,
		\Bigr\}
		\,.
	\end{equation*}
  An element $(N, \, \tau) \in \Dtry(A)$
  is called a \textbf{directory} of elements of $A$.
  We also just write
  $\tau \in \Dtry(A)$
  with $N = \dom(\tau)$.
\end{definition}

The name `directory' points to the similarity
with hierarchical file systems,
where complete paths (names) are associated to
file contents (attributes).

A directory $(N, \, \tau) \in \Dtry(A)$
can be seen as
a trie $N$,
where each leaf node $n \in N$
has an associated attribute $\tau(n) \in A$.
Hence,
the concept of named sum
carries over
from $\NamespaceSet$
to $\Dtry(A)$,
for any fixed set $A$.
Named sum can then be seen as
a flattening operation,
analogous to flattening a list of lists.
Specifically, for
the set of all lists of elements of $A$
\begin{equation*}
	\List(A)
	\: = \:
	\Bigl\{
		\,
		(n, \, f)
		\, \mid \,
		n \in \mathbb{N}, \,
		f \colon \{ 1, \, \ldots, \, n \} \to A
		\,
	\Bigr\}
	\,,
\end{equation*}
flattening is a map
$\List(\List(A)) \to \List(A)$,
which for instance sends $[[1, \, 2], \, [3, \, 4]]$
to $[1, \, 2, \, 3, \, 4]$.
Similarly,
named sum is a map
$\Dtry(\Dtry(A))$ to $\Dtry(A)$.
An input to this is
a trie $N$,
which has associated to each leaf node $n \in N$
another trie $N_n$,
whose leaves carry elements of $A$.
The named sum operation then simply
grafts onto each node $n \in N$ its associated trie $N_n$,
carrying along the elements of $A$ at the leaves.
In terms of category theory,
the flattening operation,
together with a unit operation,
makes
$\List$ and $\Dtry$
examples of monads.

\subsection{Bundle of quantities}%

The first attribute
associated to a port
represents
the physical quantity that can be exchanged via the port,
e.g.~momentum or entropy,
along with
a state space
in which values of the quantity live.
For instance
a momentum variable can take values in $\bR$
if motion along a single axis is considered,
or it can take values in the vector space $\mathfrak{se}^*(3)$
if rigid body motions in $3$-dimensional Euclidean space are considered,
etc.

\begin{definition}
	Let $\SpaceSet$ be
  a finite set of state spaces
  (i.e.~smooth manifolds)
	and for each space $\cX \in \SpaceSet$,
	let $\QuantitySet_\cX$ be
  a finite set of physical quantities
	taking values in $\cX$.
	Further,
	let $\QuantitySet$ be
  the disjoint union of all quantities, i.e.
	\begin{equation*}
		\QuantitySet
		\: = \:
		\bigsqcup_{\cX \in \SpaceSet} \QuantitySet_\cX
		\,.
	\end{equation*}
	This forms
	a \textbf{bundle of quantities}
  (in the category of finite sets),
	whose projection
	\begin{equation*}
		\mathrm{space} : \QuantitySet \twoheadrightarrow \SpaceSet
	\end{equation*}
	returns the underlying space of a quantity.
\end{definition}

\begin{example}
	The bundle of quantities given by
	\begin{equation*}
		\begin{split}
			\SpaceSet
			\: &= \:
			\{ \bR, \, \mathrm{SE}(3), \, \mathfrak{se}^*(3) \}
			\\
			\QuantitySet_\bR
			\: &= \:
			\{ \mathtt{displacement}, \, \mathtt{momentum}, \, \mathtt{entropy}, \, \mathtt{volume} \}
			\\
			\QuantitySet_{\mathrm{SE}(3)}
			\: &= \:
			\{ \mathtt{pose} \}
			\\
			\QuantitySet_{\mathfrak{se}^*(3)}
			\: &= \:
			\{ \mathtt{momentum} \}
			,
		\end{split}
	\end{equation*}
	would be sufficient for
	some spatially-lumped models
	from mechanics and thermodynamics.
	A port with associated quantity
	$q = (\bR, \, \mathtt{entropy}) \in \QuantitySet$
	can exchange entropy
	and its associated state space is
	$\mathrm{space}(q) = \bR \in \SpaceSet$.
	\label{ex:bundle_of_port_types}
\end{example}

We note that
$(\bR, \, \mathtt{energy}) \notin \QuantitySet$,
because
in the context of EPHS,
internal energy is a function of the state variable entropy,
and not vice versa,
see also~\cite{2021LohmayerKotyczkaLeyendecker}.

To distinguish
reversible and irreversible dynamics
later on in~\cref{ssec:reversible,ssec:irreversible},
we make use of the concept of \emph{parity},
see~\cite{2014PavelkaKlikaGrmela,2021LohmayerKotyczkaLeyendecker} for details.

\begin{definition}
	The \textbf{parity} of each quantity
	with respect to time-reversal transformation
	is given by a function
	$\mathrm{P} : \QuantitySet \rightarrow \{ -1, \, +1 \}$.
\end{definition}

\begin{example}
	For the bundle of quantities
	from~\cref{ex:bundle_of_port_types},
	we have
	\begin{equation*}
		\begin{alignedat}{1}
			&\mathrm{P}(\bR, \, \mathtt{displacement}) = +1
			\\
			&\mathrm{P}(\bR, \, \mathtt{momentum}) = -1
			\\
			&\mathrm{P}(\bR, \, \mathtt{entropy}) = +1
			\\
			&\mathrm{P}(\bR, \, \mathtt{volume}) = +1
		\end{alignedat}
	\end{equation*}
	because
	displacement, entropy, and volume do not change instantaneously
	when a recording
	is suddenly played backwards (positive parity),
	whereas
	momentum changes its sign (negative parity).
\end{example}

A comprehensive collection of physical quantities
with their associated parities
would be part of a software implementation.

\subsection{Interfaces as directories}%

Each port has two attributes:
the previously discussed quantity,
and a Boolean attribute
that indicates whether the port is a power port or a state port.
Indeed,
only power ports allow for an exchange of energy,
based on an exchange of the respective quantity.
State ports are required whenever
the behavior of one system
depends on the state of another system,
although no energy is exchanged between the two.

\begin{definition}
	An \textbf{interface} is a directory of elements of
	$\QuantitySet \times \{ \mathsf{p}, \mathsf{s} \}$, i.e.~let
	\begin{equation*}
		\InterfaceSet
		\: = \:
		\Dtry \bigl( \QuantitySet \times \{ \mathsf{p}, \mathsf{s} \} \bigr)
		\,.
	\end{equation*}
	Here, element $\mathsf{p}$ stands for power port
	and $\mathsf{s}$ stands for state port.
\end{definition}

\begin{example}
	Let $I = (N, \, \tau)$ be an interface with
	namespace
	$N = \{ \mathtt{p_1}, \, \mathtt{p_2}, \, \mathtt{q}\}$
	and attributes given by
	$
		\tau(\mathtt{p_1}) =
		\tau(\mathtt{p_2}) =
		\left( (\bR, \, \mathtt{momentum}), \, \mathsf{p} \right)
	$
	and
	$
		\tau(\mathtt{q}) =
		\left( (\bR, \, \mathtt{displacement}), \, \mathsf{s} \right)
	$.
	Ports $\mathtt{p_1}$ and $\mathtt{p_2}$
	are power ports
	that can exchange kinetic energy by exchanging momentum,
	along with information about the current amount of momentum.
	Port $\mathtt{q}$
	is a state port
	that can only exchange information about the current displacement.
	\label{ex:interface_lever}
\end{example}


\subsection{Bundles of port variables}%
\label{ssec:port_variables}

A state port has one port variable called state.
A power port has two additional port variables called flow and effort.
Given an interface,
the port variables of all its ports
combined
take values in a vector bundle,
called
the \emph{bundle of port variables}
associated to the interface.

\begin{definition}
	Let
	$I = (N, \, \tau)$ be an interface
	and let
	$p \in I$ be a port,
	i.e.~$p \in N$.
	\\
	The \textbf{state space associated to port} $p$ is
	\begin{equation*}
		\cX_p
		\: = \:
		(\mathrm{space} \circ \mathrm{pr}_1 \circ \tau)(p)
		\,,
	\end{equation*}
	where $\mathrm{pr}_1$ projects onto the first attribute, i.e.~the quantity.
	\\
	The \textbf{state space associated to interface} $I$ is
	given, up to isomorphism, by the Cartesian product
	\begin{equation*}
		\cX_I
		\: = \:
		\prod_{p \, \in \, I} \cX_p
		\,.
	\end{equation*}
	To define $\cX_I$ strictly,
	we assume that
	the elements of $I$ are in lexicographic order
	and that the Cartesian product preserves this order.
\end{definition}

\begin{example}
	Let
	$I_\text{gas} = (N_\text{gas}, \, \tau_\text{gas})$
	be an interface with namespace
	$N_\text{gas} = \{ \mathtt{s}, \, \mathtt{v} \}$
	and attributes given by
	$
		\tau_\text{gas}(\mathtt{s}) =
		\left( (\bR, \, \mathtt{entropy}), \, \mathsf{p} \right)
	$,
	$
		\tau_\text{gas}(\mathtt{v}) =
		\left( (\bR, \, \mathtt{volume}), \, \mathsf{p} \right)
	$.
	Then, its associated state space is
	$
		\cX_{I_\text{gas}} =
		\cX_{\mathtt{s}} \times \cX_{\mathtt{v}} =
		\bR \times \bR
	$.
	Let
	$I_\text{mass} = (N_\text{mass}, \, \tau_\text{mass})$
	be another interface with
	$N_\text{mass} = \{ \mathtt{p} \}$
	and
	$
		\tau_\text{mass}(\mathtt{p}) =
		\left( (\bR, \, \mathtt{momentum}), \, \mathsf{p} \right)
	$.
	Then,
	the named sum
	\begin{equation*}
		I
		\: = \:
		(N, \, \tau)
		\: = \:
		\sum [\mathtt{gas} \mapsto I_\text{gas}, \, \mathtt{mass} \mapsto I_\text{mass}]
	\end{equation*}
	has the namespace
	$N = \{ \mathtt{gas.s}, \, \mathtt{gas.v}, \, \mathtt{mass.p} \}$
	with attributes given by
	$
		\tau(\mathtt{gas.s}) =
		\left( (\bR, \, \mathtt{entropy}), \, \mathsf{p} \right)
	$,
	$
		\tau(\mathtt{gas.v}) =
		\left( (\bR, \, \mathtt{volume}), \, \mathsf{p} \right)
	$,
	$
		\tau(\mathtt{mass.p})) =
		\left( (\bR, \, \mathtt{momentum}), \, \mathsf{p} \right)
	$.
	We note that
	\begin{equation*}
		\cX_I
		\: = \:
		\cX_{I_\text{gas}}
		\times
		\cX_{I_\text{mass}}
		\,.
		\qedhere
	\end{equation*}
	\label{ex:interface}
\end{example}

Given any interface $I$,
each port $p \in I$
has an associated state variable
$p \mathtt{.x} \in \cX_p$.
If $p$ is a power port,
it additionally has
a flow variable $p \mathtt{.f}$
and
an effort variable $p \mathtt{.e}$.
All port variables implicitly depend on time.
As detailed in~\cref{ssec:storage},
the flow variable is the instantaneous rate of change of the state variable
and
the effort variable is the differential of a function of the state variable.
Hence, we have
$
	\left( p \mathtt{.x}, \, p \mathtt{.f}, \, p \mathtt{.e} \right)
	\in \bT \cX_p
$,
where
$
	\bT \cX_p =
	\mathrm{T} \cX_p \oplus \mathrm{T}^* \cX_p
$
denotes the direct sum (or Whitney sum)
of the respective tangent and cotangent bundles
(see, e.g.,~\cite{2012Lee}).
The flow and effort variables
are called power variables
because
their duality pairing
$\langle p \mathtt{.e} \mid p \mathtt{.f} \rangle$
yields the instantaneous power
that is
exchanged through port $p$.

Now we can formally define
the bundle of port variables associated to an interface:

\begin{definition}
	Let $I$ be an interface.
	Based on the Boolean attribute of its ports,
	we can think of its associated state space $\cX_I$ as
	the Cartesian product
	$\cX_{I,p} \times \cX_{I,s}$,
	where
	$\cX_{I,p}$ is
	the \textbf{state space of all power ports}
	and
	$\cX_{I,s}$ is
	the \textbf{state space of all state ports}.
	Accordingly,
	we let
	$
		x = (x_p, \, x_s) \in
		\cX_{I,p} \times \cX_{I,s}
	$
	denote all state variables,
	and
	we let
	$
		(f, \, e) \in
		\bT_{x_p} \cX_{I,p}
	$
	denote the flow and effort variables.
	The \textbf{bundle of port variables}
	$\mathrm{pr}_I : \cP_I \twoheadrightarrow \cX_I$
	associated to
	interface $I$
	is defined, up to isomorphism, as
	the pullback bundle of
	$\bT \cX_{I,p}$
	along the projection
	$
		\mathrm{pr}_p :
		\cX_{I,p} \times \cX_{I,s} \twoheadrightarrow
		\cX_{I,p}
	$.
	Hence,
	$\cP_I \cong  \bT \cX_{I,p} \times \cX_{I,s}$.
	\begin{equation*}
		\begin{tikzcd}[ampersand replacement=\&]
			\mlnode{
				$\bT \cX_{I,p}$
				\\
				$\color{gray} \ni (x_p, \, f, \, e)$
			}
			\&\&
			\mlnode{
				$\color{violet} \cP_I$
				\\
				$\color{gray} \ni (x, \, f, \, e)$
			}
			\\
			\mlnode{
				$\cX_{I,p}$
				\\
				$\color{gray} \ni x_p$
			}
			\&\&
			\mlnode{
				$\cX_I$
				\\
				$\color{gray} \ni x$
			}
			\&\&
			\mlnode{
				$\cX_{I,s}$
				\\
				$\color{gray} \ni x_s$
			}
			\arrow["{\mathrm{pr}_p}"', two heads, from=2-3, to=2-1]
			\arrow["{\mathrm{pr}_s}", two heads, from=2-3, to=2-5]
			\arrow[two heads, from=1-1, to=2-1]
			\arrow[two heads, color={violet}, from=1-3, to=1-1]
			\arrow["{\mathrm{pr}_\mathtt{I}}", two heads, color={violet}, from=1-3, to=2-3]
			\arrow["\lrcorner"{text={violet}, anchor=center, pos=0.125, rotate=-90}, draw=none, from=1-3, to=2-1]
		\end{tikzcd}
	\end{equation*}
	Again,
	to define the bundle of port variables strictly,
	we assume that the elements of $I$ are in lexicographic order
	and that the port variables are written in the same order as the ports.
\end{definition}

\begin{example}
	Consider again the interface from~\cref{ex:interface_lever}.
	Its bundle of port variables
	$\cP_I \twoheadrightarrow \cX_I$
	has the total space (or port space)
	\begin{equation*}
		\cP_I
		\: = \:
		\bT \bR \times \bT \bR \times \bR
		\: \ni \:
		\bigl(
		(\mathtt{p_1.x}, \, \mathtt{p_1.f}, \, \mathtt{p_1.e}), \,
		(\mathtt{p_2.x}, \, \mathtt{p_2.f}, \, \mathtt{p_2.e}), \,
		\mathtt{q.x}
		\bigl)
	\end{equation*}
	and the base space (or state space)
	\begin{equation*}
		\cX_I
		\: = \:
		\bR \times \bR \times \bR
		\: \ni \:
		\left(
		\mathtt{p_1.x}, \,
		\mathtt{p_2.x}, \,
		\mathtt{q.x}
		\right)
		.
		\qedhere
	\end{equation*}
\end{example}

\section{Syntax}%
\label{sec:interconnection-patterns}

As introduced in~\cref{ssec:syntax},
an interconnection pattern
expresses
how a composite system is formed
from a finite number of subsystems.
Syntax refers to
the entirety of such expressions and their composition,
i.e.~their hierarchical nesting.

\subsection{Interconnection patterns}%
\label{ssec:interconnection_patterns}

\begin{figure}[ht]
  \centering
  \includegraphics[width=0.9\textwidth]{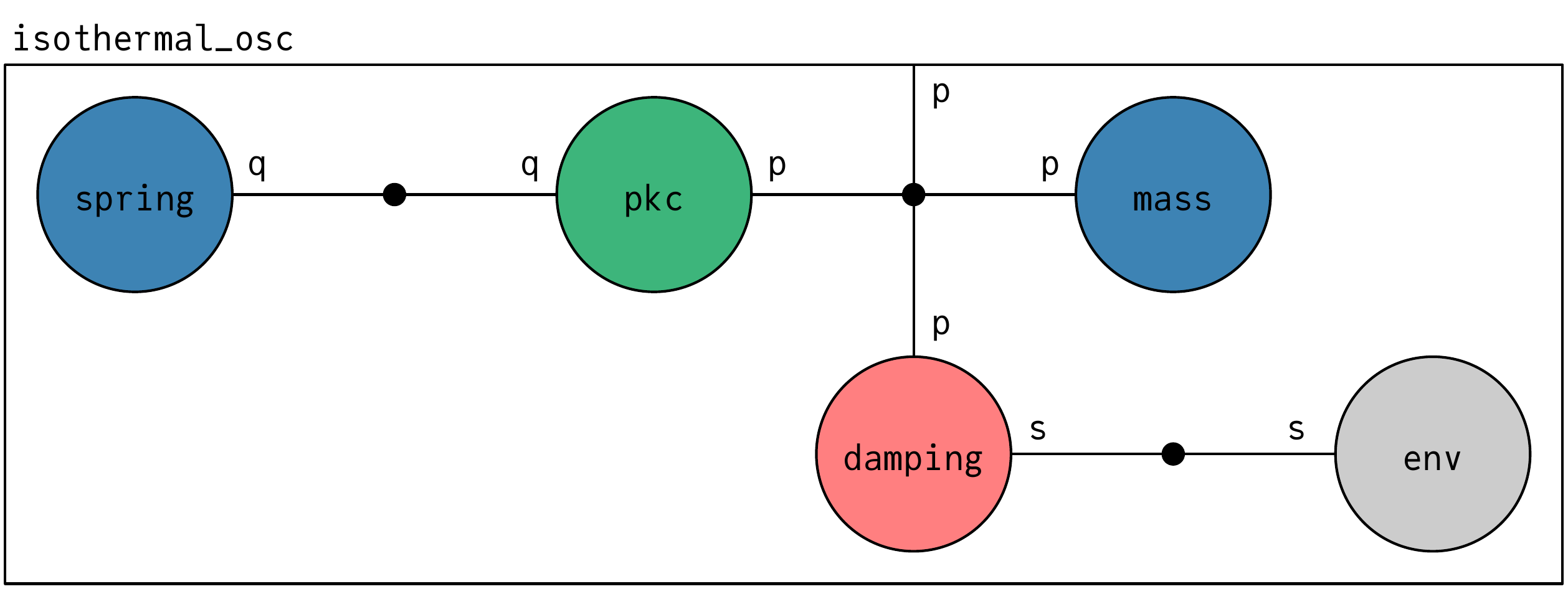}
  \caption{%
    Interconnection pattern
    of an isothermal oscillator model.
    Box $\mathtt{pkc}$ represents
    the reversible \emph{coupling} of
    the \emph{potential} energy domain
    represented by the junction on its left
    and
    the \emph{kinetic} energy domain
    represented by the junction on its right.
    Box $\mathtt{env}$ represents the
    isothermal environment,
    which absorbs the heat
    that is generated by the damping.
		The outer box represents the interface of
		the composite system.
    Its port $\mathtt{p}$
    exposes the kinetic energy domain,
    allowing for an external forcing of the oscillator model.
  }%
  \label{fig:isothermal_osc}
\end{figure}

\Cref{fig:isothermal_osc} shows
an interconnection pattern.
The mathematical content of such a diagram
is captured by the following definition.

\begin{definition}
  An \textbf{interconnection pattern}
  $(\beta, \, I, \, J)$
  is defined by
  the following data:
  \begin{enumerate}
    \item{
      A directory of interfaces
      $\beta \in \Dtry(\InterfaceSet)$.
      Names in $\dom(\beta)$ refer to \textbf{inner boxes},
      $\sum \beta$ is called \textbf{inner interface}
      and names in $\dom(\sum \beta)$ refer to \textbf{inner ports}.
    }
    \item{
      An interface $I$
      called \textbf{outer interface},
      whose ports are called \textbf{outer ports}.
    }
    \item{
      A partition $J$ of
      the combined interface
      $
      P =
      \sum \left[
        \mathtt{inner} \mapsto \sum \beta, \,
        \mathtt{outer} \mapsto I
      \right]
      $
      such that
      each subset or \textbf{junction}
      $j \in J$ contains
      \begin{itemize}
        \item{at least one inner port,}
        \item{at most one outer port,}
        \item{and all ports in $j$ have the same associated quantity. \qedhere}
      \end{itemize}
    }
  \end{enumerate}
  \label{def:pattern}
\end{definition}

To prevent name conflicts,
the definition distinguishes
inner ports and outer ports
with the prefixes $\mathtt{inner}$ and $\mathtt{outer}$.
While preventing name conflicts is crucial for
implementation of the modeling language,
we omit these prefixes
for conciseness
in what follows.

We graphically depict interconnection patterns as follows.
Each element $b \in \dom(\beta)$ is the name of an inner box,
which we draw as a circle.
The ports of its interface $\beta(b)$ are
drawn as lines emanating from the box.
The ports of the outer interface $I$ are
drawn as lines to the surrounding outer box.
Each subset $j \in J$
corresponds to a junction
drawn as a black dot.
The elements of $j \subseteq P$ are precisely the connected ports.
To avoid confusion, we recall that
the identifier \texttt{isothermal\_osc}
written on the top left of the outer box
refers to the entire pattern
(akin to a function name in programming),
rather than just to the outer box
(which similar to the return value of a function
does not need a name).

The pattern from~\cref{fig:isothermal_osc}
is redrawn in~\cref{fig:isothermal_osc_sugar},
firstly
to emphasize that
the color of the inner boxes
is merely an annotation
hinting at the nature of the considered subsystems
and secondly
to introduce an abbreviated way of writing port names.

\begin{figure}[ht]
  \centering
  \includegraphics[width=0.9\textwidth]{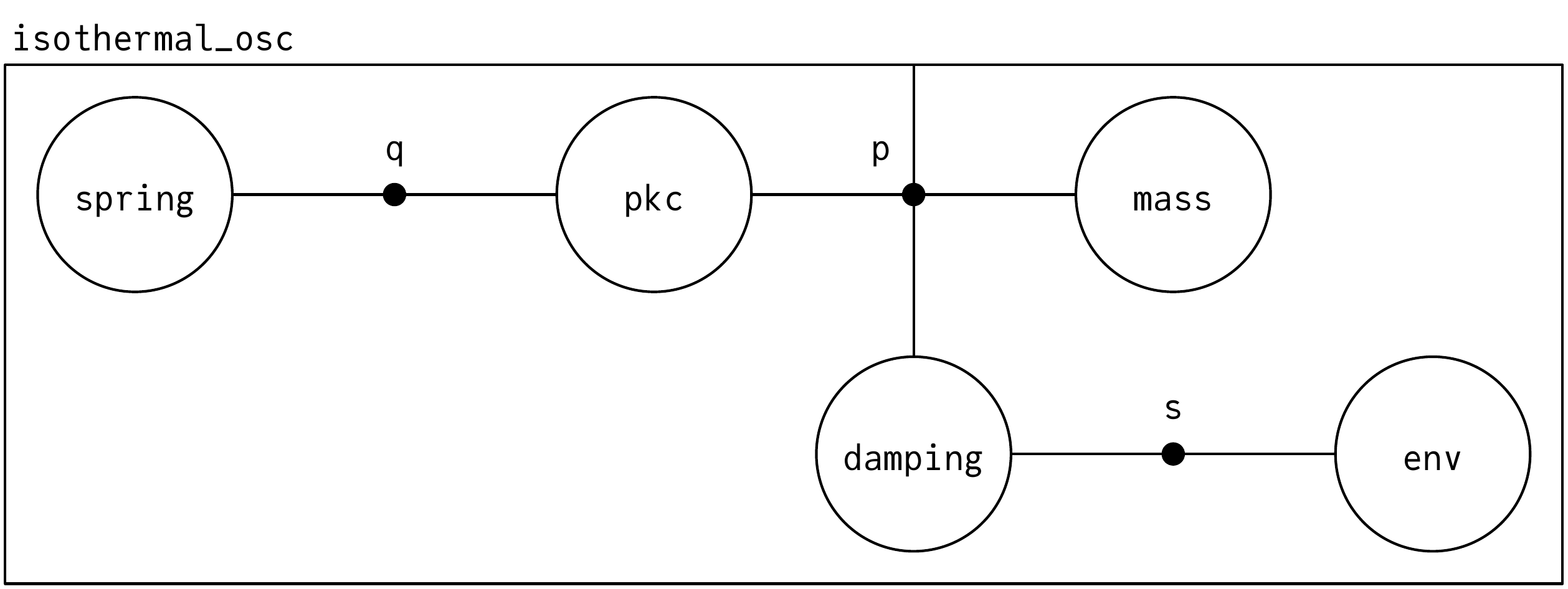}
  \caption{%
    Interconnection pattern of an isothermal oscillator model.
    Whenever all ports connected to a certain junction have the same name,
    we write the name only once at the junction.
  }%
  \label{fig:isothermal_osc_sugar}
\end{figure}

\begin{example}
  For the pattern
  $(\beta, \, I, \, J)$
  depicted in~\cref{fig:isothermal_osc_sugar},
  the namespace of inner boxes is
  \begin{equation*}
    \dom(\beta)
    \: = \:
    \{
      \mathtt{damping}, \,
      \mathtt{env}, \,
      \mathtt{mass}, \,
      \mathtt{pkc}, \,
      \mathtt{spring}
    \}
    \,.
  \end{equation*}
  The function
  $\beta : \dom(\beta) \rightarrow \InterfaceSet$
  is given by
  \begin{equation*}
    \begin{split}
      \mathtt{damping}
      \: &\mapsto \:
      (\{ \mathtt{p}, \, \mathtt{s} \}, \, \tau_\text{damping})
      \\
      \mathtt{env}
      \: &\mapsto \:
      (\{ \mathtt{s} \}, \, \tau_\text{env})
			\\
      \mathtt{mass}
      \: &\mapsto \:
      (\{ \mathtt{p} \}, \, \tau_\text{mass})
      \\
      \mathtt{pkc}
      \: &\mapsto \:
      (\{ \mathtt{p}, \, \mathtt{q} \}, \, \tau_\text{pkc})
      \\
      \mathtt{spring}
      \: &\mapsto \:
      (\{ \mathtt{q} \}, \, \tau_\text{spring})
    \end{split}
  \end{equation*}
  with
  $\tau_\text{damping}(\mathtt{p}) = ((\bR, \, \mathtt{momentum}),\, \mathsf{p})$, etc.
  Hence,
  the namespace of the inner interface $\sum \beta$ is
  \begin{equation*}
    \dom(\textstyle \sum \beta)
    \: = \:
    \{
      \mathtt{damping.p}, \,
      \mathtt{damping.s}, \,
      \mathtt{env.s}, \,
      \mathtt{mass.p}, \,
      \mathtt{pkc.p}, \,
      \mathtt{pkc.q}, \,
      \mathtt{spring.q}
    \}
    \,.
  \end{equation*}
  The namespace of the outer interface $I$ is
  $\{ \mathtt{p} \}$.
  The partition $J$ is given by
  \begin{equation*}
    \bigl\{
      \{
        \mathtt{damping.p}, \,
        \mathtt{mass.p}, \,
        \mathtt{pkc.p}, \,
        \mathtt{p}
      \}, \,
      \{
        \mathtt{pkc.q}, \,
        \mathtt{spring.q}
      \}, \,
      \{
        \mathtt{damping.s}, \,
        \mathtt{env.s}
      \}
    \bigr\}
    \,.
    \qedhere
  \end{equation*}
\end{example}

We may also denote an interconnection pattern
using the condensed notation
\begin{equation*}
  \begin{tikzcd}
    \sum \left[
      \mathtt{I_1} \mapsto I_1, \,
      \ldots, \,
      \mathtt{I_n} \mapsto I_n
    \right]
    \ar[r, twoheadrightarrow] &
    J
    \ar[r, hookleftarrow] &
    I.
  \end{tikzcd}
\end{equation*}
In this we implicitly let
$\dom(\beta) = \{ \mathtt{I_1}, \, \ldots, \, \mathtt{I_n} \}$
and
$\beta(\mathtt{I_i}) = I_i$.
The left arrow is a surjection
because each junction has at least one connected inner port
and
the right arrow is an injection
because each junction has at most one connected outer port.
Interconnection patterns are a special case of
\emph{undirected wiring diagrams}~\cite{2013Spivak},
which do not have
the injectivity and surjectivity conditions.

Multiports are just
multiple ports with a common prefix.
As demonstrated in~\cref{fig:emhd},
this enables
a \emph{syntactic sugar}
for expressing parallel connections
in a tidy and efficient manner.
We defer
the details to later work,
as the examples considered here make no use of this.

\subsection{Composition of interconnection patterns}%

Whenever
the outer interface of one pattern
and
the interface corresponding to an inner box of another pattern
are equal (up to a given renaming of the ports),
the former pattern can be substituted into the latter,
with no further data required.
This is called composition.
Having a compositional syntax
facilitates dealing with complex systems,
as it enables the encapsulation of subsystems.
These can be easily reused and replaced,
as long as interfaces match.
As an example,
the isothermal oscillator model
can be refined into
a nonisothermal oscillator model
based on the damper model in~\cref{fig:nonisothermal_damper}.

\begin{figure}[htbp]
  \centering
  \includegraphics[width=0.3\textwidth]{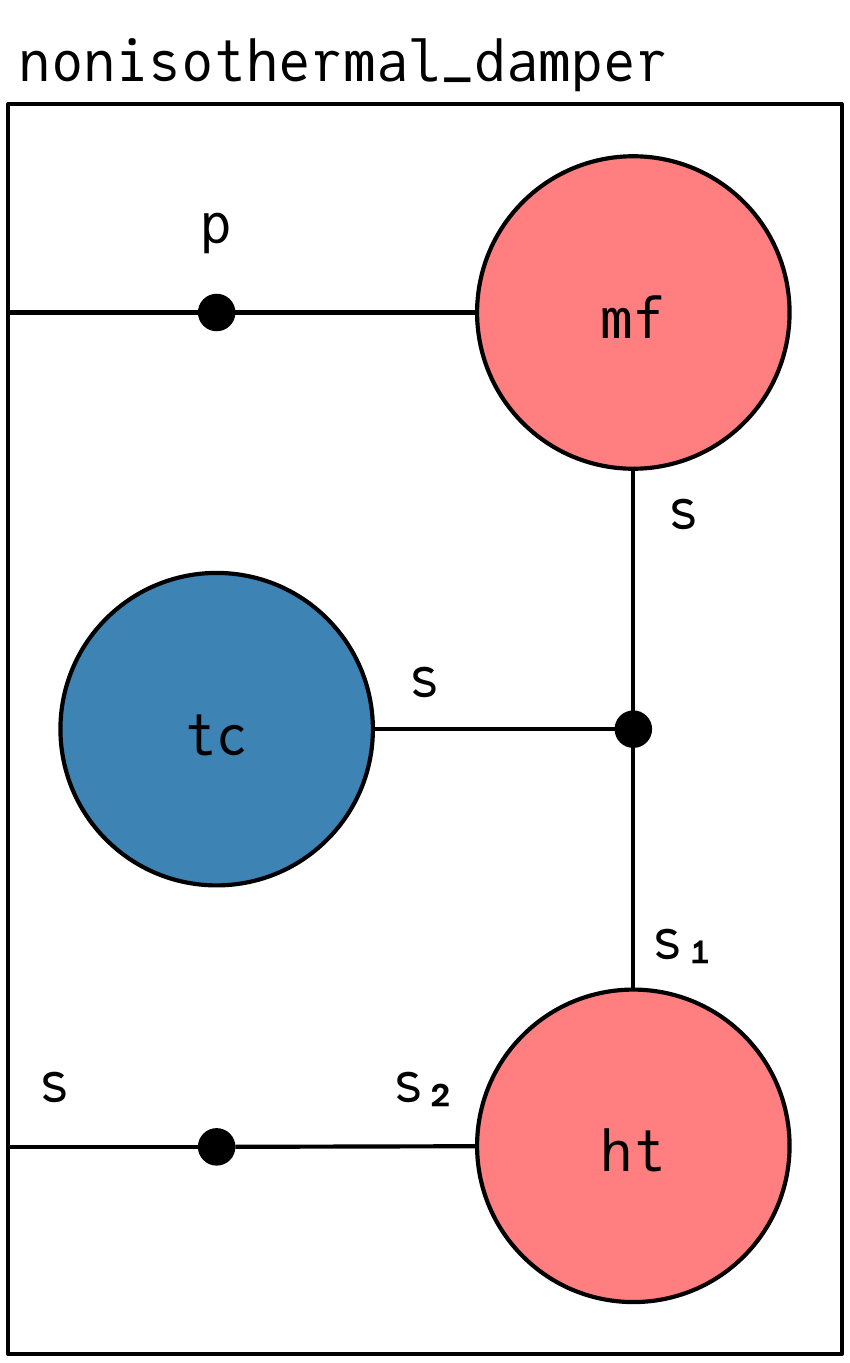}
  \caption{%
    Interconnection pattern of a nonisothermal damper model.
    Box $\mathtt{mf}$ represents
    a mechanical friction model
    connected to the kinetic energy domain
		represented by the outer port $\mathtt{p}$.
    Box $\mathtt{tc}$ represents
    a thermal capacity,
    which stores the heat generated by friction.
    Box $\mathtt{ht}$ represents
    a heat transfer model,
    which allows the thermal capacity to reach equilibrium
    with the thermal energy domain
		represented by the outer port $\mathtt{s}$.
  }%
  \label{fig:nonisothermal_damper}
\end{figure}

\begin{figure}[htbp]
  \centering
  \includegraphics[width=0.9\textwidth]{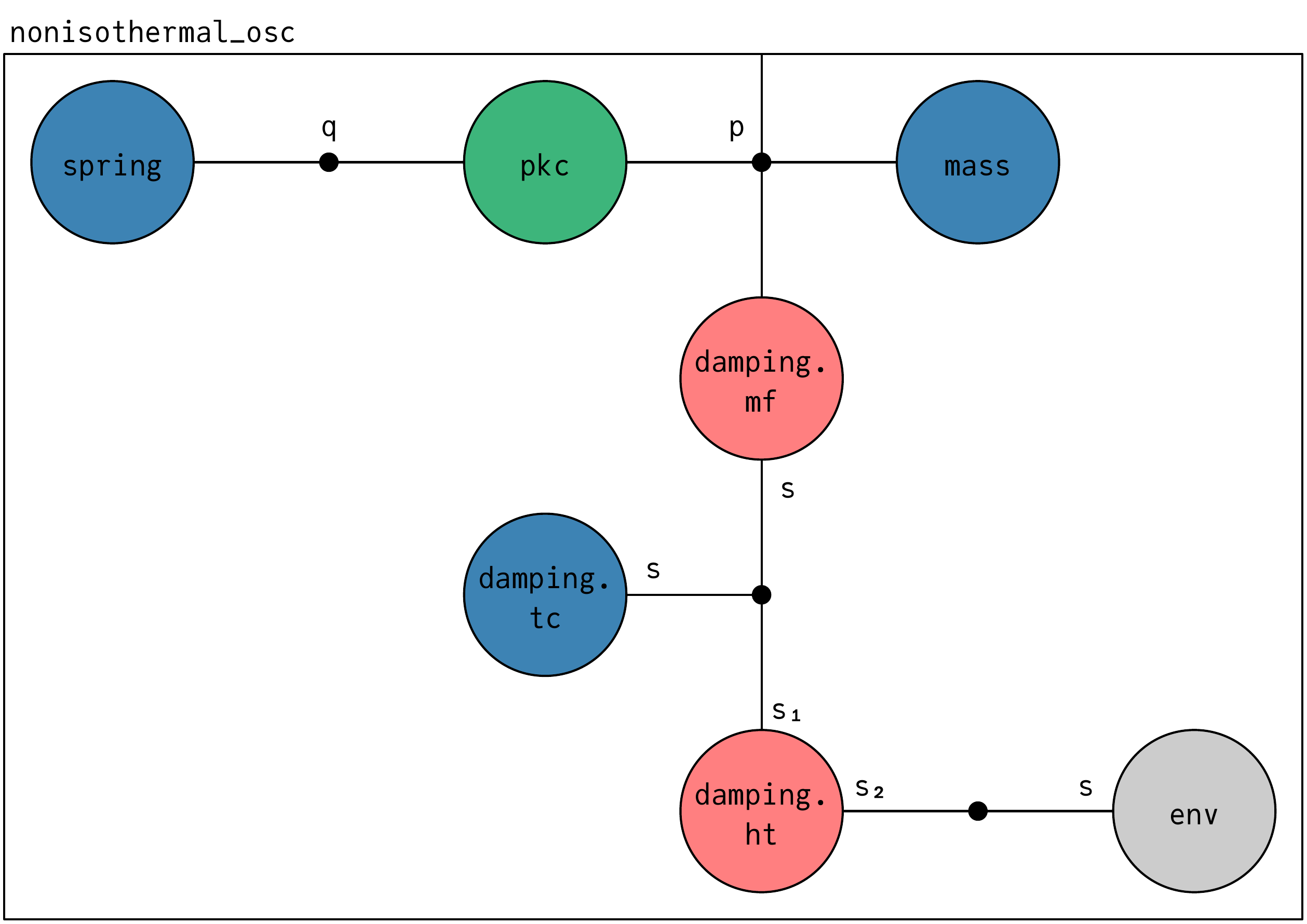}
  \caption{%
		Interconnection pattern of a nonisothermal oscillator model.
		Except for the graphical layout,
		the pattern is uniquely defined as
		the composite of the patterns in~%
		\cref{fig:isothermal_osc_sugar,fig:nonisothermal_damper}.
	}%
  \label{fig:nonisothermal_osc}
\end{figure}

Substituting the pattern in~\cref{fig:nonisothermal_damper}
into the pattern in~\cref{fig:isothermal_osc_sugar},
yields the pattern in~\cref{fig:nonisothermal_osc}.
Graphically,
the former pattern
is inserted into
the inner box called $\mathtt{damping}$ of the latter pattern.
The common interface is deleted
and for each of its ports,
the two assigned junctions (one in each pattern) are identified.
The named sum for interfaces is used
to obtain the inner interface of the composed pattern.

Category theory provides
a language to talk about composition,
and
it includes
well-behavedness conditions,
which ensure that composition works in the way we expect,
just as we need not worry about the difference between
$(1 + 2) + 3$ and $1 + (2 + 3)$.
A future paper shall formalize
the composition of interconnection patterns
based on directories and multicategories/operads.
To limit the scope of this article,
we just mention that
it is not difficult to see why
the procedure for composition
outlined in the previous paragraph
preserves the three conditions in~\cref{def:pattern}.

\section{Semantics of interconnection patterns}%
\label{sec:interconnection-semantics}

In~\cref{ssec:port_variables} we have defined
bundles of port variables associated to interfaces.
In~\cref{ssec:interconnection_patterns} we have defined
interconnection patterns,
which have an inner interface
(containing all ports of the subsystems)
and an outer interface
(containing the ports of the composite system).
We now discuss how to interpret an interconnection pattern
as imposing a relation between
the port variables of its inner interface and
the port variables of its outer interface.
We do this in two parts.
In the first part,
we talk about the relation as
a collection of equations,
and we discuss its physical interpretation.
In the second part,
we reify it as a geometric object,
namely as a relation between bundles of port variables.

\subsection{Interconnection patterns as equations}%

The semantics of an interconnection pattern
$(\beta, \, I, \, J)$
can be seen as
a collection of equations.
Specifically,
for each junction $j \in J$, we have the following.
\\
Let $j_p \subseteq j$ be the subset of connected power ports.
\begin{itemize}
  \item{%
    \textbf{Equality of state.}
    For all ports $p, p' \in j$,
    \begin{equation*}
      p \mathtt{.x}
      \: = \:
      p' \mathtt{.x}
    \end{equation*}
  }
  \item{%
    \textbf{Equality of effort.}
    For all power ports $p, p' \in j_p$,
    \begin{equation*}
      p \mathtt{.e}
      \: = \:
      p' \mathtt{.e}
    \end{equation*}
  }
  \item{%
    \textbf{Equality of net flow.}
    \begin{equation*}
      \sum_{p \, \in \, j_p \cap \sum \beta} p \mathtt{.f}
      \: = \:
      \sum_{p \, \in \, j_p \cap I} p \mathtt{.f}
    \end{equation*}
  }
\end{itemize}
We note that
the sum on the left-hand side
has at least one term (surjectivity condition)
and
the sum on the right-hand side
has at most one term (injectivity condition).

As detailed in~\cref{sec:composite-systems},
combining these equations with
equations that determine the semantics of
the subsystems filling the inner boxes
yields
a collection of equations
that determines the semantics of the resulting composite system.

Since all ports connected to a junction have the same quantity,
their associated state variables
take values in the same space,
making equality of state a well-typed equation.
As a consequence of this,
all effort variables
take values in the same fiber of the same cotangent bundle,
making equality of effort also well-typed.
Similarly,
all flow variables
take values in the same fiber of the same tangent bundle,
making equality of net flow well-typed.

Equality of state is required,
since the behavior of one system
may in general depend on the state of another system,
even if the two systems do not exchange energy.
We refer to~\cref{fig:motor,sec:example} for an example of this.

As a consequence of
equality of effort and
equality of net flow,
at every junction $j \in J$
the following power balance equation is satisfied:
\begin{equation*}
  \sum_{p \, \in \, j_p \cap \sum \beta} \langle p \mathtt{.e} \mid p \mathtt{.f} \rangle
  \: = \:
  \sum_{p \, \in \, j_p \cap I} \langle p \mathtt{.e} \mid p \mathtt{.f} \rangle
\end{equation*}
Summing over all junctions on both sides
gives that
the net power supplied to all subsystems
is equal to
the net power supplied to the composite system.
%
%

Junctions correspond to energy domains
and
so do (outer) ports,
since they expose energy domains.
Whenever two ports are connected to the same junction
the respective energy domains are identified.
This is exactly how interconnection patterns compose.
Interconnection means
that systems share energy domains.
The state of a shared energy domain
is given by the shared state variable,
according to the \emph{equality of state}.
The connected systems
may exchange the respective physical quantity
according to the \emph{equality of net flow}.
Moreover,
based on the \emph{equality of effort},
this exchange implies a transfer of (exergetic) power.

\begin{example}
  Let's again consider the pattern in~\cref{fig:nonisothermal_osc}.
  At the junction which represents the kinetic energy domain,
  the momentum variable (state) is shared and
  forces (flows) balance at equal velocity (effort):
  \begin{equation*}
    \begin{alignedat}{3}
      &\mathtt{pkc.p.x} \: {}={} \: &&\mathtt{mf.p.x} \: {}={} \: &&\mathtt{mass.p.x} \: {}={} \: \mathtt{p.x}
      \\
      &\mathtt{pkc.p.f} \: {}+{} \: &&\mathtt{mf.p.f} \: {}+{} \: &&\mathtt{mass.p.f} \: {}={} \: \mathtt{p.f}
      \\
      &\mathtt{pkc.p.e} \: {}={} \: &&\mathtt{mf.p.e} \: {}={} \: &&\mathtt{mass.p.e} \: {}={} \: \mathtt{p.e}
    \end{alignedat}
    \qedhere
  \end{equation*}
\end{example}

\Cref{tab:junctions}
summarizes the physical interpretation of junctions
representing a number of different energy domains.
\begin{table}[ht]
  \centering
  \resizebox{\columnwidth}{!}{%
    \begin{tabular}{ccc}
      \toprule
      energy domain
      &
      state variable
      &
      flows balance at equal efforts
      \\
      \midrule
      kinetic
      &
      momentum
      &
      forces balance at equal velocity
      \\
      magnetic
      &
      flux linkage
      &
      voltages balance at equal current
      \\
      potential
      &
      displacement
      &
      velocities balance at equal force
      \\
      hydraulic
      &
      volume
      &
      volume rates balance at equal pressure
      \\
      electric
      &
      charge
      &
      currents balance at equal voltage 
      \\
      thermal
      &
      entropy
      &
      entropy rates balance at equal temperature 
      \\
      \bottomrule
    \end{tabular}
  }
  \caption{Physical interpretation of junctions as energy domains.}%
  \label{tab:junctions}
\end{table}

We can now make sense of
the surjectivity and injectivity conditions
distinguishing interconnection patterns
form undirected wiring diagrams.
The surjectivity condition says that
each junction must have at least one connected inner port.
In other words,
an energy domain can only belong to a composite system
if it belongs to at least one subsystem.
The injectivity condition says that
each junction may have at most one connected outer port.
In other words,
an energy domain can be exposed at most once.
Exposing an energy domain twice
and identifying it again at a higher level
by connecting the two respective ports
would lead to undefined flow variables.

\subsection{Interconnection patterns as relations}%

Recall that
a relation $R : A \nrightarrow B$
between two sets $A$ and $B$
is given by
a subset of their Cartesian product,
i.e.~$R \subseteq A \times B$.
We use the same symbol,
here $R$,
to denote the relation and
the subobject that defines it.
The arrow symbol $\nrightarrow$ signifies that
relations have no inherent direction,
however
the formal direction is relevant to denote their composition,
and hence to think about relations as forming a category.

Analogous to
the prototypical example of a relation between two sets,
the relation assigned to an interconnection pattern
is defined by
a vector subbundle of the Cartesian product of
the bundle of port variables
associated to the inner interface and
the bundle of port variables
associated to the outer interface.

\begin{definition}
  Let $P = (\beta, \, I, \, J)$ be
  an interconnection pattern
  and
  let $I_\text{in} = \sum \beta$ denote its inner interface.
  Further,
  let $\cP_{I_\text{in}} \twoheadrightarrow \cX_{I_\text{in}}$
  be the bundle of port variables associated to the inner interface,
  and
  let $\cP_{I} \twoheadrightarrow \cX_{I}$
  be the bundle of port variables associated to the outer interface.
  Based on the partition $J$,
  let $\cP_P \subset \cP_{I_\text{in}} \times \cP_I$
  be the subspace where
  equality of state,
  equality of effort, and
  equality of net flow
  are satisfied.
  Similarly,
  let $\cX_P \subset \cX_{I_i} \times \cX_I$
  be the subspace where
  equality of state
  is satisfied.
  Then,
  the \textbf{semantics assigned to interconnection pattern} $P$
  is the relation
  $\cP_P : \cP_{I_\text{in}} \nrightarrow \cP_I$
  between the port variables of the inner interface
  and the port variables of the outer interface
  that is defined by
  the vector subbundle
  $\cP_P \twoheadrightarrow \cX_P$.
  \label{def:semantics_pattern}
\end{definition}

\begin{figure}[ht]
  \centering
  \includegraphics[width=16em]{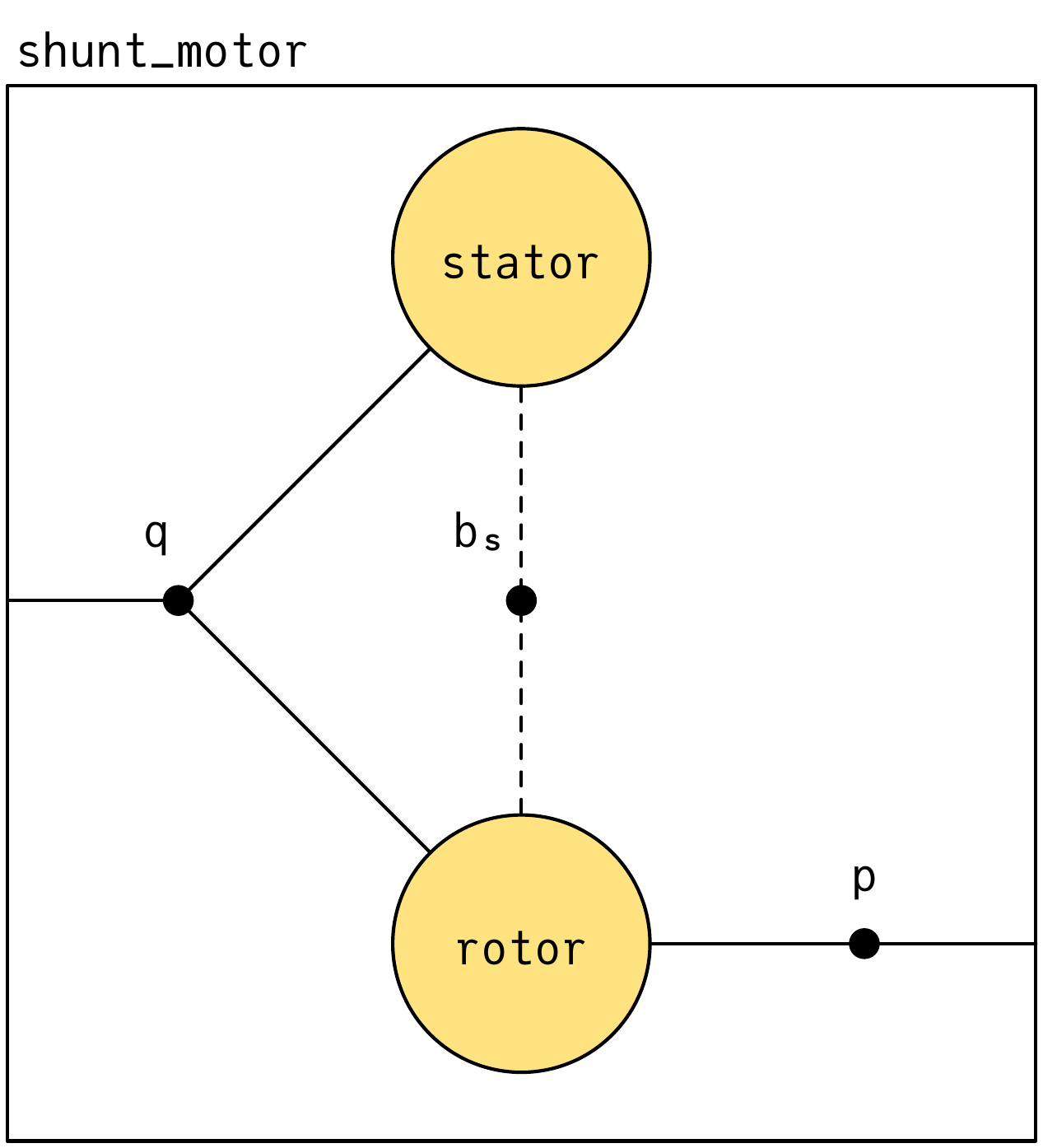}%
  \caption{%
    Interconnection pattern of the DC shunt motor model
    discussed in~\cref{sec:example}.
    The stator and rotor are connected to
    a common electric energy domain
    that is exposed via the outer port $\mathtt{q}$.
    The kinetic energy domain of the rotor
    is exposed via the outer port $\mathtt{p}$.
    The conversion between electric and kinetic energy
    in the rotor
    depends on the magnetic energy domain of the stator,
    which is represented by the junction
    with two connected state ports named $\mathtt{b_s}$.
  }%
  \label{fig:motor}
\end{figure}

\begin{example}
  Consider the interconnection pattern
  $P = (\beta, \, I, \, J)$
  depicted in~\cref{fig:motor}
  with names of inner boxes given by
  $\dom(\beta) = \{ \mathtt{rotor}, \, \mathtt{stator} \}$
  and
  interfaces given by
  \begin{equation*}
    \begin{aligned}
      &\beta(\mathtt{rotor})
      \: = \:
      \left(
        \{ \mathtt{b_s}, \, \mathtt{p}, \, \mathtt{q} \}, \,
        \tau_\text{r}
      \right)
      \\
      &\beta(\mathtt{stator})
      \: = \:
      \left(
        \{ \mathtt{b_s}, \, \mathtt{q} \}, \,
        \tau_\text{s}
      \right)
      \\
      &I
      \: = \:
      \left(
        \{ \mathtt{p}, \, \mathtt{q} \}, \,
        \tau
      \right)
      \,,
    \end{aligned}
  \end{equation*}
  where
  \begin{equation*}
    \begin{aligned}
      &\tau_\text{r}(\mathtt{b_s})
      \: = \:
			\tau_\text{s}(\mathtt{b_s})
      \: = \:
      ((\bR, \, \mathtt{flux\_linkage}), \, \mathsf{s})
      \\
      &\tau_\text{r}(\mathtt{p})
      \: = \:
      \tau(\mathtt{p})
      \: = \:
      ((\bR, \, \mathtt{angular\_momentum}), \, \mathsf{p})
      \\
      &\tau_\text{r}(\mathtt{q})
      \: = \:
			\tau_\text{s}(\mathtt{q})
      \: = \:
      \tau(\mathtt{q})
      \: = \:
      ((\bR, \, \mathtt{charge}), \, \mathsf{p})
      \,.
    \end{aligned}
  \end{equation*}
  Further,
  the interconnection is determined by
  \begin{equation*}
    J
    \: = \:
    \bigl\{
      \{ \mathtt{rotor.b_s}, \, \mathtt{stator.b_s} \}, \,
      \{ \mathtt{rotor.p}, \, \mathtt{p} \}, \,
      \{ \mathtt{rotor.q}, \, \mathtt{stator.q}, \, \mathtt{q} \}
    \bigr\}
    \,.
  \end{equation*}
  The port space of the inner interface
  $I_\text{in} = \sum \beta$ is
  \begin{equation*}
    \begin{split}
      \cP_{I_\text{in}}
      \: = \:
      &\bR \times
			\bT \bR \times
      \bT \bR \times
      \bR \times
      \bT \bR
      \\
      \: \ni \:
      &\bigl(
        \mathtt{rotor.b_s.x}, \,
        (\mathtt{rotor.p.x}, \, \mathtt{rotor.p.f}, \, \mathtt{rotor.p.e}), \,
        \\
        &\:\:
        (\mathtt{rotor.q.x}, \, \mathtt{rotor.q.f}, \, \mathtt{rotor.q.e}), \,
        \\
        &\:\:
        \mathtt{stator.b_s.x}, \,
        (\mathtt{stator.q.x}, \, \mathtt{stator.q.f}, \, \mathtt{stator.q.e})
      \bigr)
    \end{split}
  \end{equation*}
  and
  the port space of the outer interface $I$ is
  \begin{equation*}
    \cP_I
    \: = \:
    \bT \bR \times
    \bT \bR
    \: \ni \:
    \bigl(
      (\mathtt{p.x}, \, \mathtt{p.f}, \, \mathtt{p.e}), \,
      (\mathtt{q.x}, \, \mathtt{q.f}, \, \mathtt{q.e})
    \bigr)
    \,.
  \end{equation*}
  Finally,
  the subspace
  $\cP_P \subset \cP_{I_\text{in}} \times \cP_I$
  defining the relation is
  \begin{equation*}
    \begin{split}
      \cP_P
      \: = \:
      \bigl\{
        &\bigl(
          \mathtt{rotor.b_s.x}
          \, , \,
					(\mathtt{rotor.p.x}, \, \mathtt{rotor.p.f}, \, \mathtt{rotor.p.e})
          \, , \,
          \\
          &(\mathtt{rotor.q.x}, \, \mathtt{rotor.q.f}, \, \mathtt{rotor.q.e})
          \, , \,
          \\
          &\mathtt{stator.b_s.x}
          \, , \,
					(\mathtt{stator.q.x}, \, \mathtt{stator.q.f}, \, \mathtt{stator.q.e})
          \, , \,
          \\
          &(\mathtt{p.x}, \, \mathtt{p.f}, \, \mathtt{p.e})
          \, , \,
					(\mathtt{q.x}, \, \mathtt{q.f}, \, \mathtt{q.e})
        \bigr)
        \in \cP_{I_i} \times \cP_{I}
        \, \mid \,
        \\
        &\mathtt{rotor.b_s.x} = \mathtt{stator.b_s.x}
        \, , \,
        \\
        &\mathtt{rotor.p.x} = \mathtt{p.x}
        \, , \,
        \\
        &\mathtt{rotor.p.f} = \mathtt{p.f}
        \, , \,
        \\
        &\mathtt{rotor.p.e} = \mathtt{p.e}
        \, , \,
        \\
        &\mathtt{rotor.q.x} = \mathtt{stator.q.x} = \mathtt{q.x}
        \, , \,
        \\
        &\mathtt{rotor.q.f} + \mathtt{stator.q.f} = \mathtt{q.f}
        \, , \,
        \\
        &\mathtt{rotor.q.e} = \mathtt{stator.q.e} = \mathtt{q.e}
        \,
      \bigr\}
      \,.
      \qedhere
    \end{split}
  \end{equation*}%
  \label{ex:motor_pattern}
\end{example}

Given a hierarchy of interconnection patterns,
we can first compose the patterns
and then determine the relation assigned to the composed pattern.
Alternatively,
we can first determine the relation assigned to each individual pattern
and then compose these relations.
Both ways lead to the same result.
The combinatorial syntax
and its relational semantics
hence compose in compatible ways.
In the language of category theory,
we say that the (assignment of) semantics is \emph{functorial}.
A future paper shall use
directories and the theory of multicategories/operads
to demonstrate that~\cref{def:semantics_pattern} indeed
defines a suitable notion of functor.

\section{Exergy}%
\label{sec:exergy}

Every system is ultimately determined by
a power-preserving interconnection of primitive subsystems,
called components.
These represent elementary physical behaviors, namely
storage, and reversible as well as irreversible exchange of energy.
To achieve a thermodynamically consistent
combination of reversible and irreversible dynamics,
components are defined
with respect to
an exergy reference environment.
Before components can be defined in the next section,
in this section we define the reference environment.
We start with its physical interpretation.

\subsection{Intuition}%

As briefly discussed in~\cref{ssec:thermodynamics},
thermodynamic models
have an irreversible dynamics,
since not all physically-relevant degrees of freedom
are determined by the state variables.
According to the first and second law of thermodynamics,
internal energy,
which is given by
an `imprecise' phenomenological description,
cannot be fully converted into
mechanical or electromagnetic energy,
which is `precise',
in the sense that
all relevant degrees of freedom
are fully determined by the state variables.

The irreversible degradation of
`mechanically precise' forms of energy
into
`thermodynamically imprecise' forms of energy
can be quantified
using the concept of exergy.
A hypothetical conversion device,
whose mode of operation is only restricted by
the first and second law of thermodynamics
is used to
determine how much mechanical (or electromagnetic) energy
can be extracted from a system,
before it reaches thermodynamic equilibrium
with a reference environment,
see~\cref{fig:exergy}.
The environment is defined by
intensive variables,
such as temperature and pressure,
which are constant in time and space.

\begin{figure}[ht]
  \centering
  \includegraphics[width=0.6\textwidth]{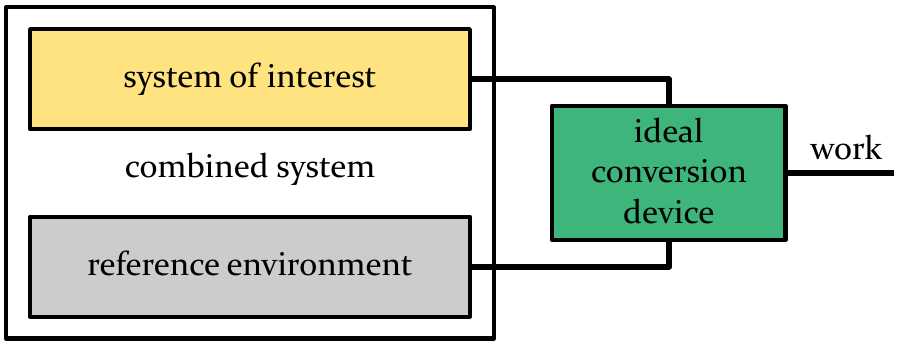}
  \caption{
    To assess the exergy content of a system,
    it is combined with a reference environment.
    The exergy content is then defined as
    the maximum amount of work
    that can be extracted
    from the combined system
    via a hypothetical conversion device
    that operates completely reversibly.
  }%
  \label{fig:exergy}
\end{figure}

The ideal conversion device
is \textit{not} restricted to
interacting with the system through its interface.
Instead,
the device operates in a counterfactual
where all energy domains
of both the system and the environment can be accessed.
The exergy content of a composite system
is hence obtained by
(recursively) summing
the exergy stored in each of its storage components.
The exergy content of
a component storing
purely mechanical or electromagnetic forms of energy
is simply equal to its energy content.
For components storing internal energy,
the exergy content is assessed based on
the ideal conversion device
and the reference environment.
A concrete statement of this
follows in~\cref{ex:gas},
after the definition of
the exergy storage function in~\cref{def:exergy}.


For more information about
exergy,
its history,
and its links to port-Hamiltonian as well as metriplectic or GENERIC systems,
we refer to~\cite{2021LohmayerKotyczkaLeyendecker}.

\subsection{Reference environment}%

The internal energy of a thermodynamic system
is a function of its entropy,
and possibly further extensive quantities
such as volume, mass of chemical species, etc.
To assess the exergy content of the system,
we need for each such variable
an energy-conjugate, intensive variable
attributed to the reference environment.
Since these intensive variables are constant,
they determine the equilibrium state
of the combined system,
see~\cref{fig:exergy}.

\begin{definition}
  A \textbf{reference environment}
  $(I_R, \, \lambda_R)$
  is defined by
  an interface
  $I_R = (N_R, \, \tau_R)$
  whose ports
  are all power ports
  and have
  a unique extensive quantity
  in $\QuantitySet_\bR$.
  One port must have the quantity
  $(\bR, \, \mathtt{entropy})$.
  Moreover,
  each port $p \in I_R$ has
  an associated value
  $\lambda_R(p) \in \bR^*$
  of the energy-conjugate intensive quantity.
\end{definition}

Uniqueness of quantities,
i.e.~injectivity of
$\tau_R : N_R \rightarrow \mathbf{Q}_\bR \times \{ \mathsf{p} \}$,
is required to
determine the exergy storage function of a system
from its energy storage function.
Regarding the physical interpretation,
each port corresponds to
an energy domain that is present in the environment.
Allowing multiple energy domains of the same kind
would conflict with
the environment always being in equilibrium with itself.

\begin{example}
  The minimal definition of a reference environment
  represents an isothermal heat bath.
  This amounts to
  an interface with a single entropy port
  and an associated reference temperature,
  henceforth denoted by
  ${\color{violet} \theta_0}$:
  \begin{equation*}
    \begin{alignedat}{1}
      N_R
      \: &= \:
      \{ \mathtt{s} \}
      \\
      \tau_R(\mathtt{s})
      \: &= \:
      \left( (\bR, \, \mathtt{entropy}), \, \mathsf{p} \right)
      \\
      \lambda_R(\mathtt{s})
      \: &= \:
      {\color{violet} \theta_0}
      \,.
    \end{alignedat}
  \end{equation*}
  To determine the exergy content of
  pressurized fluids
  with respect to an isobaric atmosphere,
  the definition of the environment is extended to
  $N_R = \{ \mathtt{s}, \, \mathtt{v} \}$
  with
  \begin{equation*}
    \begin{alignedat}{1}
      \tau_R(\mathtt{v})
      \: &= \:
      \left( (\bR, \, \mathtt{volume}), \, \mathsf{p} \right)
      \\
      \lambda_R(\mathtt{v})
      \: &= \:
      {\color{violet} -\pi_0}
      \,.
    \end{alignedat}
  \end{equation*}
  The minus sign in front of the reference pressure
  ${\color{violet} \pi_0}$
  stems from the fact that
  more volume means less energy.
  To model for instance
  a mass diffusion process,
  a port $\mathtt{m}$
  representing
  the mass of a chemical species
  with a constant chemical potential
  ${\color{violet} \mu_0}$
  can be included,
  i.e.~let
  $N_R = \{ \mathtt{s}, \, \mathtt{v}, \, \mathtt{m} \}$
  with
  \begin{equation*}
    \begin{alignedat}{1}
      \tau_R(\mathtt{m})
      \: &= \:
      \left( (\bR, \, \mathtt{mass}), \, \mathsf{p} \right)
      \\
      \lambda_R(\mathtt{m})
      \: &= \:
      {\color{violet} \mu_0}
      \,.
    \end{alignedat}
    \qedhere
  \end{equation*}
\end{example}

A software implementation should already include
a comprehensive definition of a reference environment.
Allowing the user to change the values of the intensive quantities
would be of interest for
exergy analysis and thermodynamic optimization.

\section{Primitive systems}%
\label{sec:primitive-systems}

At the lowest level,
any hierarchically-structured model
is defined by primitive systems,
called components.
There are four kinds of them:
\begin{itemize}
  \item{%
    \textbf{Storage components} represent
    energy storage.
  }
  \item{%
    \textbf{Reversible components} represent
    reversible energy exchange
    \\ 
    or a transformation between different forms of
    describing an energy domain.
  }
  \item{%
    \textbf{Irreversible components} represent
    irreversible energy exchange.
  }
  \item{%
    \textbf{Environment components} may be used to
    model energy exchange with
    \\ 
    the reference environment.
  }
\end{itemize}

Components are in general defined with respect to
a reference environment
$(I_R, \, \lambda_R)$.
Further,
every component has an interface
$I = (N, \, \tau)$
with port variables
$(x, \, f, \, e) \in \cP_I$.
To define the exergy storage function
of storage components
and to
state conditions that
reversible and irreversible components have to satisfy,
we define the following notation:
For any port $p \in I$,
let
\begin{equation*}
  \lambda(p)
  \: = \:
  \begin{cases}
    \lambda_R(p_R)
    &\exists \, p_R \in I_R
    \, \colon \,
    \tau_R(p_R)
    \: = \:
    \tau(p)
    \\
    0 \in \mathrm{T}^* \cX_p
    &\text{else.}
  \end{cases}
\end{equation*}
With regard to the existential quantifier,
this is well-defined due to
the injectivity of $\tau_R$.
For instance,
if $p$ represents a thermal energy domain then
$\lambda(p) = {\color{violet} \theta_0}$,
whereas
if $p$ represents a kinetic energy domain then
$\lambda(p) = 0$,
since the environment does not have a kinetic energy domain.
We extend the above notation
to an entire interface:
Given for example
$N = \{ \mathtt{p_1}, \, \ldots, \, \mathtt{p_n} \}$,
let
\begin{equation*}
  \lambda(I)
  \: = \:
  \bigl(
    \lambda(\mathtt{p_1}), \,
    \ldots, \,
    \lambda(\mathtt{p_n})
  \bigr)
  \,.
\end{equation*}
We also need a similar notation
that references a specific environment port:
For any port $p \in I$ and $p_R \in I_R$,
let
\begin{equation*}
  \lambda(p, \, p_R)
  \: = \:
  \begin{cases}
    \lambda_R(p_R)
    &\tau_R(p_R)
    \: = \:
    \tau(p)
    \\
    0 \in \mathrm{T}^* \cX_p
    &\text{else.}
  \end{cases}
\end{equation*}
Again,
given for example
$N = \{ \mathtt{p_1}, \, \ldots, \, \mathtt{p_n} \}$,
let
\begin{equation*}
  \lambda(I, \, p_R)
  \: = \:
  \bigl(
    \lambda(\mathtt{p_1}, \, p_R), \,
    \ldots, \,
    \lambda(\mathtt{p_n}, \, p_R)
  \bigr)
  \,.
\end{equation*}

\subsection{Storage components}%
\label{ssec:storage}

Storage components
are primitive systems
that represent energy storage.
Examples include
elastic springs, movable masses,
capacitors, inductors,
and thermal capacities.

\begin{definition}
  A \textbf{storage component}
  $(I, \, E)$
  is defined by
  an interface $I$
  containing only power ports
  and
  an energy storage function
  $E : \cX_I \rightarrow \bR$.
\end{definition}

The function $E$
yields the stored energy
for any state $x \in \cX_I$.
Based on the reference environment $I_R$,
the exergy content is determined as follows.

\begin{definition}
  Let $(I, \, E)$ be a storage component.
  Its \textbf{exergy storage function}
  $H : \cX_I \rightarrow \bR$
  is defined by
  \begin{equation*}
    H(x)
    \: = \:
    E(x) -
    \sum_{p \, \in \, I}
    \sum_{\substack{%
      p_R \, \in \, I_R \\
      \tau_R(p_R) \, = \, \tau(p)
    }}
    \langle \lambda_R(p_R) \mid p \mathtt{.x} \rangle
    \,.
    \qedhere
  \end{equation*}
  \label{def:exergy}
\end{definition}

The exergy function $H$
follows directly from the energy function $E$
by subtracting the pairing between
each state variable
representing an extensive quantity
also present in the environment,
and the constant, energy-conjugate quantity
from the environment.
We come back to this in~\cref{ex:gas}.


By definition,
the total exergy of a (composite) system
is zero when it reaches equilibrium
in itself and with the environment.
However,
the equilibrium state
of the considered composite system
is usually unknown
when defining a (reusable) component.
Regarding the dynamics,
this is irrelevant,
as it depends only on the differential $\dd H$.
In case it would matter for thermodynamic analysis,
the contribution of the component to the total exergy
is given by $H(x) - H(x_0)$,
where $x_0$ is the state of the component
when the considered composite system reached equilibrium.

The effort variables are given by
the differential of the exergy function
evaluated at the current state
$x \in \cX_I$, i.e.
\begin{equation*}
  \forall \, p \in I
  \: : \:
  p \mathtt{.e}
  \: = \:
  \left. \pdv{H}{(p \mathtt{.x})} \right\vert_{x}
  \,.
\end{equation*}
The evolution of the state variables
is obtained by integrating the flow variables,
i.e.
\begin{equation*}
  \forall \, p \in I
  \: : \:
  \dv{(p \mathtt{.x})}{t}
  \: = \:
  p \mathtt{.f}
  \,.
\end{equation*}
Based on the chain rule,
the rate of change of the stored exergy
is given by
$
\dv{H}{t} =
\sum_{p \in I} \:
\langle p \mathtt{.e} \mid p \mathtt{.f} \rangle
$.

\begin{example}
  Consider
  a storage component
  $(I_\text{gas}, \, E_\text{gas})$
  representing
  a gas-filled compartment
  whose entropy and volume can change.
  Its interface
  $I_\text{gas}$
  is defined in~\cref{ex:interface},
  and we assume
  knowledge of its energy storage function
  $E_\text{gas} : \bR \times \bR \rightarrow \bR$.
  With
  $x = (s, \, v)$,
  the exergy storage function
  is then given by
  \begin{equation*}
    H_\text{gas}(s, \, v)
    \: = \:
    E_\text{gas}(s, \, v)
    \, - \, {\color{violet} \theta_0} \, s
    \, + \, {\color{violet} \pi_0} \, v
    \,.
  \end{equation*}
  As
  the temperature is defined by
  $\theta = \pdv{E_\text{gas}}{s}$
  and the (negative) pressure is defined by
  $-\pi = \pdv{E_\text{gas}}{v}$,
  the effort variables are
  given by
  \begin{equation*}
    \begin{split}
      \mathtt{s.e}
      \: &= \:
      \theta - {\color{violet} \theta_0}
      \\
      \mathtt{v.e}
      \: &= \:
      -(\pi - {\color{violet} \pi_0})
      \,.
    \end{split}
  \end{equation*}
  Since the entropy rate
  $\dot{s} = \mathtt{s.f}$
  is related to
  the heat flow rate
  $\Phi$
  via
  $\dot{s} = \frac{1}{\theta} \, \Phi$,
  the exergetic power
  at the entropy port is
  $
  \langle \mathtt{s.e} \mid \mathtt{s.f} \rangle
  =
  \frac{\theta - {\color{violet} \theta_0}}{\theta} \, \Phi
  =
  \eta_\text{Carnot} \, \Phi
  $,
  where
  $\eta_\text{Carnot}$
  is the efficiency
  of an ideal heat engine
  operating between
  a hot source with temperature $\theta$
  and
  a cold sink with temperature ${\color{violet} \theta_0}$.
  Considering a piston with
  the pressurized gas on one side
  and the atmosphere on the other,
  $
  \langle \mathtt{v.e} \mid \mathtt{v.f} \rangle
  =
  -(\pi - {\color{violet} \pi_0}) \, \dot{v}
  $
  is the power that can be extracted
  from the pressurized gas
  when the volume expansion rate is $\dot{v}$.
  \label{ex:gas}
\end{example}

Next,
we formally state
the semantics of a storage component
with interface $I$
in terms of a relation
$\cS : \mathrm{T} \cX_I \nrightarrow \cP_I$.
This is used in~\cref{sec:composite-systems},
where we define the semantics of a composite system
in terms of
the semantics of its primitive subsystems
and
the semantics of
its (hierarchy of nested) interconnection pattern(s).

\begin{definition}
  The \textbf{semantics assigned to storage component} $(I, \, E)$
  is the relation
  $\cS : \mathrm{T} \cX_I \nrightarrow \cP_I$
  that is given by the subbundle
  \begin{equation*}
    \cS
    \: = \:
    \Bigl\{
			\,
      \bigl( (\tilde{x}, \, \tilde{f}), \, (x, \, f, \, e) \bigr)
      \in \mathrm{T} \cX_I \times \cP_I
      \, \mid \,
      \tilde{x} \: = \: x, \,
      \tilde{f} \: = \: f, \,
      e \: = \: \dd H(x)
			\,
    \Bigr\}
    \,,
  \end{equation*}
  where $H$ is induced from $E$
  according to~\cref{def:exergy}.
  \label{def:semantics_storage}
\end{definition}

\subsection{Environment components}%

From a user perspective,
environment components provide
a simple alternative to storage components,
when dealing with low-grade energy.%
They may represent
an isothermal heat bath,
which absorbs heat
generated by irreversible components,
or an isobaric atmosphere,
which exchanges pressure-volume work
with other components, etc.
Environment components are essentially
storage components
with a trivial energy and exergy function.

\begin{definition}
  An \textbf{environment component} is defined by
  a subinterface
  $I \subseteq I_R$
  containing the exposed energy domains
  of the reference environment.
\end{definition}

The state variables
keep track of
the exchanged entropy, volume, etc.
As for storage components,
their rate of change
is given by
the flow variables, i.e.
\begin{equation*}
  \forall \, p \in I \: : \:
  \dv{(p \mathtt{.x})}{t}
  \: = \:
  p \mathtt{.f}
  \,.
\end{equation*}
By definition, the environment has zero exergy content.
Since the effort variables are given by the differential of
the exergy function, they are also zero:
\begin{equation*}
  \forall \, p \in I \: : \:
  p \mathtt{.e}
  \: = \:
  0
  \,.
\end{equation*}
In agreement with the definition,
the chain rule also tells us that
the rate of change of the stored exergy is zero.


\begin{definition}
  The \textbf{semantics assigned to environment component}
  $I \subseteq I_R$
  is the relation
  \\ 
  $\cE : \mathrm{T} \cX_I \nrightarrow \cP_I$
  that is given by the subbundle
  \begin{equation*}
    \cE
    \: = \:
    \Bigl\{
			\,
      \bigl( (\tilde{x}, \, \tilde{f}), \, (x, \, f, \, e) \bigr)
      \in \mathrm{T} \cX_I \times \cP_I
      \, \mid \,
      \tilde{x} \: = \: x, \,
      \tilde{f} \: = \: f, \,
      e \: = \: 0
			\,
    \bigr\}
    \,.
    \qedhere
  \end{equation*}
  \label{def:semantics_environment}
\end{definition}

\subsection{Reversible components}%
\label{ssec:reversible}

In a generalized sense,
reversible components represent
gyrators,
transformers,
and/or kinematic constraints.
Generalized gyrators may describe
a coupling of kinetic and potential energy domains
in mechanics,
or a coupling of magnetic and electric energy domains
in electromagnetism.
Examples
where one domain reversibly interacts with itself are
gyroscopic forces in rigid body dynamics
and
self-advection of kinetic energy in fluid dynamics.
The Lorentz force is
an example
where a kinetic energy domain interacts with itself
in a way that depends on
the state of a magnetic energy domain.
Generalized transformers may for instance describe
ideal mechanical levers
or
electrical transformers.
Frame transformations in multibody systems
are an example relating
different representations
of the same energy domain.
Examples of
kinematic constraints include
joint constraints in multibody dynamics
and
the incompressibility constraint in fluid dynamics.
Further,
the parallel connection of capacitors,
the series connection of inductors or springs,
and rigid links between point masses
are described by constraints.

The geometric structure behind reversible components
is called Dirac structure.
For our purposes,
a Dirac structure $\cD$ is
a power-preserving, locally linear relation of port variables.
More precisely,
we could define
$\cD \subset \cP_I$
as a maximally isotropic vector subbundle
of the bundle of port variables
$\cP_I \twoheadrightarrow \cX_I$,
see~\cite{1990Courant,2014SchaftJeltsema,2013Bursztyn}.
However,
to state additional conditions that reversible components have to satisfy,
we give a more practical definition
of Dirac structures
in terms of a particular representation,
which might be called
\emph{constrained hybrid input-output representation}
(cf.~\cite{2014SchaftJeltsema}).

\begin{definition}
  Let $I$ be an interface
  and let
  $
  \cX_{I,p}
  \cong \cX_1 \times \cX_2
  \ni x_p = (x_1, \, x_2)
  $
  be a factorization of
  the state space of all power ports
  with flow and effort variables correspondingly denoted by
  $f = (f_1, \, f_2)$
  and
  $e = (e_1, \, e_2)$.
  A \textbf{Dirac structure} $\cD \subset \cP_I$
  admits the following representation:
  \begin{equation*}
    \cD
    \, = \,
    \left\{
      (x, \, f, \, e)
      \in \cP_I
      \, \mid \,
      \exists \, \lambda_c \in V^*
      \, \colon
      \left[
        \begin{array}{c}
          f_1
          \\ \hline
          e_2
          \\ \hline
          0
        \end{array}
      \right]
      \, = \,
      \left[
        \begin{array}{c | c | c}
          L(x) & -g(x) & C^*(x)
          \\ \hline
          g^*(x) & 0 & 0
          \\ \hline
          -C(x) & 0 & 0
        \end{array}
      \right]
      \,
      \left[
        \begin{array}{c}
          e_1
          \\ \hline
          f_2
          \\ \hline
          \lambda_c
        \end{array}
      \right]
    \right\}
  \end{equation*}
  Here,
  $
  L :
  \mathrm{T}^* \cX_1 \rightarrow
  \mathrm{T} \cX_1
  $,
  $
  g :
  \mathrm{T} \cX_2 \rightarrow
  \mathrm{T} \cX_1
  $,
  its linear dual
  $
  g^* :
  \mathrm{T}^* \cX_1 \rightarrow
  \mathrm{T}^* \cX_2
  $,
  \\ 
  $
  C :
  \mathrm{T}^* \cX_1 \rightarrow
  \cX_I \times V
  $,
  and its dual
  $
  C^* :
  \cX_I \times V^* \rightarrow
  \mathrm{T} \cX_1
  $
  are vector bundle maps
  over the identity on $\cX_I$.
  Further,
  $V$ is a vector space
  and
  $L(x)$ is skew-symmetric
  for all $x \in \cX_I$.%
	\label{def:dirac}
\end{definition}

The maps $L$, $g$ and $C$ respectively define
a generalized gyrator,
a generalized transformer and
a constraint.
The constraint is enforced via the Lagrange multiplier $\lambda_c$.
Due to
the skew-symmetry of
the matrix in~\cref{def:dirac},
exergetic power is conserved, i.e.
\begin{equation*}
  \langle e_1 \mid f_1 \rangle
  \: + \:
  \langle e_2 \mid f_2 \rangle
  \: = \:
  0
  \,.
\end{equation*}

We remark that
this actually defines an \emph{almost} Dirac structure,
as no integrability condition,
akin to the Jacobi identity for Poisson structures,
is required.
However,
in the unconstrained or holonomically constrained case,
such a condition can be checked.

\begin{example}
  Let the interface
  $I_\text{hkc} = (N_\text{hkc}, \, \tau_\text{hkc})$
  be given by
  \begin{equation*}
    \begin{split}
      N_\text{hkc}
      \: &= \:
      \{ \mathtt{v_1}, \, \mathtt{v_2}, \, \mathtt{p} \}
      \\
      \tau_\text{hkc}(\mathtt{v_1})
      \: &= \:
      \tau_\text{hkc}(\mathtt{v_2})
      \: = \:
      ((\bR, \, \mathtt{volume}), \, \mathsf{p})
      \\
      \tau_\text{hkc}(\mathtt{p})
      \: &= \:
      ((\bR, \, \mathtt{momentum}), \, \mathsf{p})
    \end{split}
  \end{equation*}
  and let $a \in \bR$
  be a constant surface area of a piston.
  Then,
  the equation
  \begin{equation*}
    \left[
      \begin{array}{c}
        \mathtt{v_1.f} \\
        \mathtt{v_2.f} \\
        \mathtt{p.f}
      \end{array}
    \right]
    \: = \:
    \left[
      \begin{array}{rrr}
        0  &  0 & -a \\
        0  &  0 & +a \\
        +a & -a &  0
      \end{array}
    \right]
    \,
    \left[
      \begin{array}{c}
        \mathtt{v_1.e} \\
        \mathtt{v_2.e} \\
        \mathtt{p.e}
      \end{array}
    \right]
  \end{equation*}
  defines a Dirac structure
  $\cD_\text{hkc} \subset \cP_{I_\text{hkc}}$.
  This describes a generalized gyrator
  that couples the two hydraulic energy domains
  at the front and the back side of a piston
  with the kinetic energy domain
  of the piston's moving mass.
  \label{ex:hkc_dirac}
\end{example}

A reversible component is defined by a Dirac structure
that satisfies two additional conditions
ensuring thermodynamic consistency.

\begin{definition}
  A \textbf{reversible component}
  $(I, \, \cD)$
  is given by
  an interface $I$
  and
  a Dirac structure
  $\cD \subset \cP_I$
  such that
  all quantities present in the reference environment
	(in particular entropy)
	are conserved
  and
  the time-reversal invariance property holds.
  The first condition requires that
	$\forall \, p_R \in I_R$
	and
	$\forall \, x \in \cX_I \, \colon$
  \begin{equation*}
      0
      \: = \:
      L(x) \, \lambda(I_1, \, p_R)
      \quad \text{and} \quad
      0
      \: = \:
      g^*(x) \, \lambda(I_1, \, p_R)
      \quad \text{and} \quad
      0
      \: = \:
      C(x) \, \lambda(I_1, \, p_R)
    \,.
  \end{equation*}
  The time-reversal invariance property holds if
  $\forall \, x \in \cX_I \, \colon$
  \begin{equation*}
    \begin{alignedat}{3}
      {(L(x))}^{ij} \neq 0 &\implies
      -&&\mathrm{P}(x_1^i)
      \: &&= \:
      \mathrm{P}({(L(x))}^{ij}) \, \mathrm{P}(x_1^j)
      \\
      {(g(x))}^{ik} \neq 0 &\implies
      &&\mathrm{P}(x_1^i)
      \: &&= \:
      \mathrm{P}({(g(x))}^{ik}) \, \mathrm{P}(x_2^k)
      \,.
    \end{alignedat}
  \end{equation*}%
  Here, we use abstract index notation
  without implied summation.
  Indices $i$ and $j$ range over
  the power ports associated with $\cX_1$
  and $k$ ranges over
  the power ports associated with $\cX_2$.
  The state variable is accordingly split as
  $x = (x_p, \, x_s) = ((x_1, \, x_2), \, x_s)$
  and
  $\mathrm{P}$ yields the parity of the respective quantity
  with
  $\mathrm{P}(\mathrm{const.}) = +1$.
  \label{def:reversible}
\end{definition}

According to the chain rule,
a Dirac structure $\cD$
conserves the quantity
corresponding to some port $p_R \in I_R$
if any resulting flow variable $f_1$
satisfies
$0 = \langle \lambda(I_1, \, p_R) \mid f_1 \rangle$.
Inserting the expression for $f_1$,
this condition is satisfied if
$0 = \langle \lambda(I_1, \, p_R) \mid L(x) \, e_1 \rangle$
is satisfied for any effort variable $e_1$,
and
$0 = \langle \lambda(I_1, \, p_R) \mid g(x) \, f_2 \rangle$
is satisfied for any flow variable $f_2$,
and
$0 = \langle \lambda(I_1, \, p_R) \mid C^*(x) \, \lambda_c \rangle$
is satisfied for any Lagrange multiplier $\lambda_c$.
Taking the linear adjoint (transpose),
the first condition requires that
$0 = \langle L^*(x) \, \lambda(I_1, \, p_R) \mid e_1 \rangle$
for any $e_1$.
This simplifies to the first part of the condition in the definition,
since $L$ is skew-adjoint.
The other two parts follow analogously.
Regarding details on the time-reversal invariance condition,
we refer to~\cite{2014PavelkaKlikaGrmela}
or~\cite{2021LohmayerKotyczkaLeyendecker}.

\begin{example}
  The reversible component
  $(I_\text{hkc}, \, \cD_\text{hkc})$,
  as already defined in~\cref{ex:hkc_dirac},
  describes
  the coupling between
  the kinetic energy domain
  of a piston with surface area $a$
  and the two hydraulic energy domains on either side of the piston.
  The first condition here ensures that volume is conserved.
  It is satisfied since
  \begin{equation*}
    \left[
      \begin{array}{c}
        0 \\
        0 \\
        0
      \end{array}
    \right]
    \: = \:
    \left[
      \begin{array}{rrr}
        0  &  0 & -a \\
        0  &  0 & +a \\
        +a & -a &  0
      \end{array}
    \right]
    \,
    \left[
      \begin{array}{c}
        {\color{violet} -\pi_0} \\
        {\color{violet} -\pi_0} \\
        0
      \end{array}
    \right]
    \,.
  \end{equation*}
  The time-reversal invariance property holds since
  \begin{equation*}
    \begin{split}
      -\mathrm{P}(\mathtt{(\bR, \, volume}))
      \: &= \:
      \mathrm{P}(-a) \, \mathrm{P}((\bR, \, \mathtt{momentum}))
      \\
      -\mathrm{P}((\bR, \, \mathtt{volume}))
      \: &= \:
      \mathrm{P}(+a) \, \mathrm{P}((\bR, \, \mathtt{momentum}))
      \,.
    \end{split}
  \end{equation*}
  Checks for
  the other two non-zero entries of $L$
  pass due to skew-symmetry of $L$.
\end{example}

The semantics of a reversible component is simply given by
its Dirac structure
interpreted as a relation.

\begin{definition}
  The \textbf{semantics assigned to reversible component}
  $(I, \, \cD)$
  is the relation
  $\cD : 1 \nrightarrow \cP_I$
  that is given by its Dirac structure $\cD$.
  \label{def:semantics_reversible}
\end{definition}

The symbol $1$
in $1 \nrightarrow \cP_I$
refers to
the Cartesian unit for vector bundles,
i.e.~the vector bundle with
a single point as its base space
and the zero-dimensional vector space above it.

\subsection{Irreversible components}%
\label{ssec:irreversible}

Irreversible components are
phenomenological models of
irreversible processes
with local thermodynamic equilibrium.
Examples include
mechanical friction,
fluid viscosity,
electrical conduction / resistance,
thermal conduction,
and mass diffusion.
Such models
can be defined in terms of
a function that maps
thermodynamic forces,
which drive the relaxation process,
to
thermodynamic fluxes,
which seek to restore equilibrium.
Based on the assumption that
the microscopic dynamics are
reversible,
Onsager found that
such maps must satisfy a symmetry property~%
\cite{1931Onsager}.
This motivates our definition of
an~\emph{Onsager structure}.

\begin{definition}
  Let $I$ be an interface.
  An \textbf{Onsager structure} $\cO \subset \cP_I$
  is a fiber subbundle that
  admits the following representation:
  \begin{equation*}
    \cO
    \, = \,
    \left\{
      (x, \, f, \, e)
      \in \cP_I
      \, \mid \,
      f
      \: = \:
      \frac{1}{\color{violet} \theta_0} \,
      M(e) \, e
    \right\}
    \,,
  \end{equation*}
  where
  $M$ is a smooth function
  that yields
  a symmetric non-negative definite linear operator
  $M(e) \colon \mathrm{T}^* \cX_{I,p} \to \mathrm{T} \cX_{I,p}$
  for all effort variables
  $e \in \mathrm{T}^* \cX_{I,p}$.
  \label{def:onsager_structure}
\end{definition}

In contrast to a Dirac structure,
an Onsager structure is
neither power-preserving,
nor is it a locally linear relation
of the port variables,
since destruction of exergy,
or equivalently production of entropy,
is usually described
by a non-linear relation.
Hence,
$\cO \subset \cP_I$
is not a vector subbundle of $\cP_I$,
but merely a fiber subbundle.
The symmetry of $M(e)$
corresponds to
Onsager's reciprocal relations
and
its non-negative definiteness
implies a non-negative exergy destruction rate,
i.e.~$\langle e \mid f \rangle \geq 0$.

\begin{example}
  Let the interface
  $I_\text{mf} = (N_\text{mf}, \, \tau_\text{mf})$
  be given by
  \begin{equation*}
    \begin{split}
      N_\text{mf}
      \: &= \:
      \{ \mathtt{p}, \, \mathtt{s} \}
      \\
      \tau_\text{mf}(\mathtt{p})
      \: &= \:
      ((\bR, \, \mathtt{momentum}), \, \mathsf{p})
      \\
      \tau_\text{mf}(\mathtt{s})
      \: &= \:
      ((\bR, \, \mathtt{entropy}), \, \mathsf{p})
    \end{split}
  \end{equation*}
  and let $d \in \bR_{>0}$
  be a friction coefficient.
  Then,
  the following
  defines an Onsager structure
  $\cO_\text{mf} \subset \cP_{I_\text{mf}}$
  representing mechanical friction:
  \begin{equation*}
    \begin{bmatrix}
      \mathtt{p.f} \\
      \mathtt{s.f}
    \end{bmatrix}
    \: = \:
    \frac{1}{\color{violet} \theta_0} \,
    d \,
    \begin{bmatrix}
      \theta & -\upsilon \\
      -\upsilon & \frac{\upsilon^2}{\theta}
    \end{bmatrix}
    \,
    \begin{bmatrix}
      \mathtt{p.e} \\
      \mathtt{s.e}
    \end{bmatrix}
  \end{equation*}
  Here,
  $\upsilon = \mathtt{p.e}$
  is the velocity
  and
  $\theta = {\color{violet} \theta_0} + \mathtt{s.e}$
  is the absolute temperature
  at which kinetic energy is dissipated
  into the thermal energy domain,
	since adding ${\color{violet} \theta_0}$
	cancels the shift $-{\color{violet} \theta_0}$
	in the effort variable,
	see~\cref{ex:gas}.
  The phenomenological parameter $d$
  could also be a function of
  the effort variables $e$,
  i.e.~it could depend on
  the velocity and the temperature,
  or theoretically also the base point $x$.
  The above can be simplified to
  $\mathtt{p.f} = d \, \upsilon$
  and
  $\mathtt{s.f} = -\frac{1}{\theta} \, d \, \upsilon^2$.
  Hence,
  the environment temperature ${\color{violet} \theta_0}$
  cancels out
  in the flow variables,
  which here are
  the friction force
  and
  the entropy production rate
  (with a minus sign since
  thermal energy is leaving the component).
	The exergy destruction rate is given by
	$
	\langle \mathtt{p.e} \mid \mathtt{p.f} \rangle +
	\langle \mathtt{s.e} \mid \mathtt{s.f} \rangle =
	{\color{violet} \theta_0} \, \frac{d \, \upsilon^2}{\theta} \geq 0
	$,
	with $\frac{d \, \upsilon^2}{\theta} \geq 0$
	being the entropy production rate.
  \label{ex:mf_onsager}
\end{example}

An irreversible component is defined by an Onsager structure
that satisfies two additional conditions
ensuring thermodynamic consistency.

\begin{definition}
  An \textbf{irreversible component}
  $(I, \, \cO)$
  is given by
  an interface $I$
  and
  an Onsager structure
  $\cO \subset \cP_I$
  such that
  energy and all quantities of the reference environment,
  except for entropy,
  are conserved
  and
  the time-reversal invariance property does strictly not hold.
  The condition for
  conservation of energy
  requires that
  $\forall \, e \in \mathrm{T}^* \cX_{I,p} \, \colon$
  \begin{equation*}
    0
    \: = \:
    M(e) \, \bigl( \lambda(I) + e \bigr)
    \,.
  \end{equation*}
  The condition that
  all quantities present in the reference environment,
  except for entropy,
  are conserved
  requires that
  $\forall \, p_R \in I_R$
	with
  $\tau_R(p_R) \neq ((\bR, \, \mathtt{entropy}), \, \mathsf{p})$,
	we have that
	$\forall \, e \in \mathrm{T}^* \cX_{I,p} \, \colon$
  \begin{equation*}
		0
		\: = \:
		M(e) \, \lambda(I, \, p_R)
    \,.
  \end{equation*}
  The time-reversal invariance property
  does strictly not hold if
  $\forall \, x \in \cX_I \, \colon$
  \begin{equation*}
    {(M(x))}^{ij} \neq 0 \implies
    \mathrm{P}(x^i)
    \: = \:
    \mathrm{P}({(M(x))}^{ij}) \, \mathrm{P}(x^j)
    \,.
  \end{equation*}
  Here, we use the same notation as
  in~\cref{def:reversible}
  with
  indices $i$ and $j$ ranging over
  all power ports of $I$.
  \label{def:irreversible}
\end{definition}

From~\cref{def:semantics_storage} and \cref{ex:gas}
we know that efforts variables are given by
the differential of the energy function
and additional shifts
by the constant, intensive quantities of the environment,
whenever the quantity of the respective port
matches with a quantity present in the envirionment.
We thus have
$e = e^\prime - \lambda(I)$,
where $e^\prime$ is precisely the part
that stems from the differential of the energy.
From this and the chain rule,
it follows that
energy is conserved if
$
0
= \langle e^\prime \mid f \rangle
= \langle e + \lambda(I) \mid M(e) \, e \rangle
$.
Taking the linear adjoint,
this results in the condition in the definition,
since $M(e)$ is symmetric.
Regarding the condition that
all quantities that are also present in the environment,
except for entropy, are conserved,
follows as already shown for
reversible components,
following~\cref{def:reversible}.
Again,
regarding the time-reversal invariance condition,
we refer to~\cite{2014PavelkaKlikaGrmela}
or~\cite{2021LohmayerKotyczkaLeyendecker}.

\begin{example}
  The irreversible component
  $(I_\text{mf}, \, \cO_\text{mf})$,
  as already defined in~\cref{ex:mf_onsager},
  models mechanical friction.
  The condition for conservation of energy
  is satisfied since
  \begin{equation*}
    \begin{bmatrix}
      0 \\
      0
    \end{bmatrix}
    \: = \:
    d \,
    \begin{bmatrix}
      \theta & -\upsilon \\
      -\upsilon & \frac{\upsilon^2}{\theta}
    \end{bmatrix}
    \,
    \begin{bmatrix}
      \upsilon \\
      \theta
    \end{bmatrix}
    \,.
  \end{equation*}
  The condition that
  all quantities of the reference environment,
  except for entropy,
  are conserved
  is trivially satisfied.
	As required,
  the time-reversal invariance property
  does strictly not hold since
  \begin{equation*}
    \begin{split}
      \mathrm{P}((\bR, \, \mathtt{momentum}))
      \: &= \:
      \mathrm{P}(d) \,
      \mathrm{P}(\theta) \,
      \mathrm{P}((\bR, \, \mathtt{momentum}))
      \\
      \mathrm{P}((\bR, \, \mathtt{momentum}))
      \: &= \:
      \mathrm{P}(d) \,
      \mathrm{P}(-\upsilon) \, \mathrm{P}((\bR, \, \mathtt{entropy}))
      \\
      \mathrm{P}((\bR, \, \mathtt{entropy}))
      \: &= \:
      \mathrm{P}(d) \,
      \mathrm{P}(-\upsilon) \, \mathrm{P}((\bR, \, \mathtt{momentum}))
      \\
      \mathrm{P}((\bR, \, \mathtt{entropy}))
      \: &= \:
      \mathrm{P}(d) \,
      {(\mathrm{P}(\upsilon))}^2 \, / \, \mathrm{P}(\theta) \,
      \mathrm{P}((\bR, \, \mathtt{entropy}))
      \,.
    \end{split}
  \end{equation*}
  Here, we used that
  temperature has positive parity
  and
  velocity has negative parity,
	which follows
	from the parities of entropy and momentum,
	knowing that energy has positive parity.
  \label{ex:mf}
\end{example}

Irreversible components,
and hence EPHS in general,
respect
Onsager's reciprocal relations
because the symmetry of
$M(e)$ in~\cref{def:onsager_structure}
implies equal cross-effect coefficients
(see, e.g.,~\cite{2016MielkeRengerPeletier}).
In the terminology of
\textit{Linear Irreversible Thermodynamics} (LIT),
thermodynamic fluxes
are related to thermodynamic forces
by a symmetric operator.
The following example shows this in detail.

\begin{example}
  Let the interface
  $I_\text{ht} = (N_\text{ht}, \, \tau_\text{ht})$
  be given by
  \begin{equation*}
    \begin{split}
      N_\text{ht}
      \: &= \:
      \{ \mathtt{s_1}, \, \mathtt{s_2} \}
      \\
      \tau_\text{ht}(\mathtt{s_1})
      \: &= \:
      \tau_\text{ht}(\mathtt{s_2})
      \: = \:
      ((\bR, \, \mathtt{entropy}), \, \mathsf{p})
      \,,
    \end{split}
  \end{equation*}
  let $\alpha \in \bR_{>0}$
  be a heat transfer coefficient,
  and
  let the Onsager structure
  $\cO_\text{ht} \subset \cP_{I_\text{ht}}$
  be given by
  \begin{equation*}
    \begin{bmatrix}
      \mathtt{s_1.f} \\
      \mathtt{s_2.f}
    \end{bmatrix}
    \: = \:
    \frac{1}{\color{violet} \theta_0} \,
    \alpha \,
    \begin{bmatrix}
      \frac{\theta_2}{\theta_1} & -1 \\
      -1 & \frac{\theta_1}{\theta_2}
    \end{bmatrix}
    \,
    \begin{bmatrix}
      \mathtt{s_1.e} \\
      \mathtt{s_2.e}
    \end{bmatrix}
    \,,
  \end{equation*}
  where
  $\theta_1 = {\color{violet} \theta_0} + \mathtt{s_1.e}$
  and
  $\theta_2 = {\color{violet} \theta_0} + \mathtt{s_2.e}$
  are the absolute temperatures of the two thermal energy domains.
  It can be easily checked that
  $(I_\text{ht}, \, \cO_\text{ht})$
  defines an irreversible component.
  Again,
  the phenomenological parameter $\alpha$
  could for instance be a function of the temperatures.
  The above can be simplified to
  $\mathtt{s_1.f} = -\frac{1}{\theta_1} \, \alpha \, (\theta_2 - \theta_1)$
  and
  $\mathtt{s_2.f} = -\frac{1}{\theta_2} \, \alpha \, (\theta_1 - \theta_2)$.
  Hence,
  the environment temperature ${\color{violet} \theta_0}$
  cancels out
  in the flow variables,
  which are
  rates of entropy change due to heat transfer.
	Since $M(e)$ is symmetric and non-negative definite,
	an Onsager structure can be represented in a factorized form:
	\begin{equation*}
		\underbrace{%
		\begin{bmatrix}
			\mathtt{s_1.f} \\
			\mathtt{s_2.f}
		\end{bmatrix}
		}_{f}
		\: = \:
		\underbrace{%
		\begin{bmatrix}
			+\frac{1}{\theta_1} \\
			-\frac{1}{\theta_3}
		\end{bmatrix}
		}_{C(e)}
		\,
		\underbrace{%
		\begin{bmatrix}
			\alpha \,
			\theta_1 \, \theta_3
		\end{bmatrix}
		}_{D(e)}
		\frac{1}{{\color{violet} \theta_0}} \,
		\underbrace{%
		\begin{bmatrix}
			+\frac{1}{\theta_1} &
			-\frac{1}{\theta_3}
		\end{bmatrix}
		}_{C^*(e)}
		\,
		\underbrace{%
		\begin{bmatrix}
			\mathtt{s_1.e} \\
			\mathtt{s_2.e}
		\end{bmatrix}
		}_{e}
	\end{equation*}
	The factorization has a physical interpretation in terms of LIT
	(see, e.g.,~\cite{2005Oettinger}).
	The thermodynamic force
	that drives the irreversible process is given by
	\begin{equation*}
		X
		\: = \:
		\frac{1}{{\color{violet} \theta_0}} \,
		C^*(e) \, e
		\: = \:
		\frac{1}{\theta_2} - \frac{1}{\theta_1}
		\,.
	\end{equation*}
	Application of the symmetric operator $D(e)$
	yields the thermodynamic flux
	\begin{equation*}
		J
		\: = \:
		D(e) \,
		\frac{1}{{\color{violet} \theta_0}} \,
		C^*(e) \, e
		\: = \:
		\alpha \, ( \theta_1 - \theta_2)
		\,.
	\end{equation*}
	Since $M(e)$ models a single process,
	$D(e)$ has rank $1$.
	Application of $C(e)$
	yields the flow variables $f$.
	The net power
	$
	\langle \mathtt{s_1.e} \mid \mathtt{s_1.f} \rangle +
	\langle \mathtt{s_2.e} \mid \mathtt{s_2.f} \rangle =
	\langle e \mid f \rangle =
	{\color{violet} \theta_0} \,
	\frac{\alpha \, {(\theta_1 - \theta_2)}^2}{\theta_1 \, \theta_2}
	\geq 0
	$
	is the exergy destruction rate,
	and
	$
	\frac{\alpha \, {(\theta_1 - \theta_2)}^2}{\theta_1 \, \theta_2}
	= \langle X \mid J \rangle
	\geq 0
	$
	is the entropy production rate.
  \label{ex:ht}
\end{example}

The semantics of an irreversible component is simply given by
its Onsager structure
interpreted as a relation.

\begin{definition}
  The \textbf{semantics assigned to irreversible component}
  $(I, \, \cO)$
  is the relation
  $\cO : 1 \nrightarrow \cP_I$
  that is given by its Onsager structure $\cO$.
  \label{def:semantics_irreversible}
\end{definition}

\section{Composite systems}%
\label{sec:composite-systems}

A composite system is formed by
interconnecting
a finite number of (primitive or composite) systems
according to an interconnection pattern.
The data defining a composite system consists of
an interconnection pattern and,
for each of its inner boxes,
a system with the corresponding interface.

\begin{definition}
  Let $\SystemSet$ refer to
  the \textbf{set of all systems}
  defined with respect to a fixed reference environment.
  A \textbf{composite system}
  $(I, \, \gamma, \, J)$
  is defined by the following data:
  \begin{enumerate}
    \item{%
      An interface $I$ of the composite system.
    }
    \item{%
      A directory of systems $\gamma \in \Dtry(\SystemSet)$,
      which can naturally be reduced to
      a directory of interfaces $\beta \in \Dtry(\InterfaceSet)$.
    }
    \item{
      A partition $J$ of
      the combined interface
      $
      P =
      \sum \left[
        \mathtt{inner} \mapsto \sum \beta, \,
        \mathtt{outer} \mapsto I
      \right]
      $
      such that
      $(\beta, \, I, \, J)$ defines an interconnection pattern.
      \qedhere
    }
  \end{enumerate}
\end{definition}

\begin{table}[hbt]
  \centering
  {
    \setlength{\tabcolsep}{1pt}
    \begin{tabular}{c@{\hskip 20pt}c@{\hskip 20pt}cccccc@{\hskip 3pt}c@{\hskip 20pt}c}
      \toprule
      class of systems &
      interface &
      \multicolumn{7}{c}{relation} &
      definition
      \\
      \midrule
      storage component &
      $I_s$ &
      $\cS$ & $\colon$ & $\mathrm{T} \cX_{I_s}$ & $\times$ & $1$ & $\nrightarrow$ & $\cP_{I_s}$ &
      \ref{def:semantics_storage}
      \\
      environment component &
      $I_e$ &
      $\cE$ & $\colon$ & $1$ & $\times$ & $\mathrm{T} \cX_{I_e}$ & $\nrightarrow$ & $\cP_{I_e}$ &
      \ref{def:semantics_environment}
      \\
      reversible component &
      $I_r$ &
      $\cD$ & $\colon$ & $1$ & $\times$ & $1$ & $\nrightarrow$ & $\cP_{I_r}$ &
      \ref{def:semantics_reversible}
      \\
      irreversible component &
      $I_i$ &
      $\cO$ & $\colon$ & $1$ & $\times$ & $1$ & $\nrightarrow$ & $\cP_{I_i}$ &
      \ref{def:semantics_irreversible}
      \\
      composite system &
      $I$ &
      $\cC$ & $\colon$ & $\mathrm{T} \cX_{I_s}$ & $\times$ & $\mathrm{T} \cX_{I_e}$ & $\nrightarrow$ & $\cP_{I}$ &
      \ref{def:semantics_composite}
      \\
      \bottomrule
    \end{tabular}
  }
  \caption{%
    Overview of systems
    and their semantics given by a relation.
  }%
  \label{tab:systems}
\end{table}

\Cref{tab:systems} recalls
the different kinds of components
defined in~\cref{sec:primitive-systems}.
In order to
write the relations that determine their semantics
in a unified form,
the Cartesian unit is inserted a few times.
As shown in the last row of the table,
the semantics of a composite system
takes the same form.
We just need to compose
the relation assigned to the interconnection pattern
with the Cartesian product of
the relations assigned to the subsystems
filling its inner boxes.

\begin{definition}
  Let $(I, \, \gamma, \, J)$ be a composite system
  and
  let $P = (\beta, \, I, \, J)$ be
  its underlying interconnection pattern.
  Further, let
  $\cP_P$ be the relation assigned to pattern $P$
  according to~\cref{def:semantics_pattern}.
	For any filled inner box $b \in \dom(\gamma)$,
  let $\semantic{s}$ denote the relation
  assigned to the subsystem $s = \gamma(b)$.
  In case system $s$ is itself a composite system,
  the present definition applies recursively.
  Then,
  the \textbf{semantics assigned to composite system}
  $(I, \, \gamma, \, J)$
  is the relation
  $
  \cC :
  \mathrm{T} \cX_{I_s}
  \times
  \mathrm{T} \cX_{I_e}
  \nrightarrow \cP_I$
  given by
  \begin{equation*}
    \cC
    \: = \:
    \cP_P \circ
    \left(
      \prod_{b \, \in \, \dom(\gamma)} \semantic{\gamma(b)}
    \right)
    \,.
  \end{equation*}
  Here,
  $\cX_{I_s}$ denotes the state space of
  all storage components
  and
  $\cX_{I_e}$ denotes the state space of
  all environment components
  that are subsystems of the composite system
  in a recursive sense.
  \label{def:semantics_composite}
\end{definition}

\begin{example}
  Consider the composite system
  given by filling
  the inner boxes of
  the interconnection pattern
  for the nonisothermal damper
  shown in~\cref{fig:nonisothermal_damper}
  with three suitable systems.
  The system filling the box $\mathtt{tc}$
  represents a thermal capacity
  and
  it is defined analogous to~\cref{ex:gas},
  but without a port for volume exchange.
  It gives a relation
  \begin{equation*}
    \cS_\text{tc} : \mathrm{T} \cX_\text{tc} \times 1 \nrightarrow \bT \cX_\text{tc}
  \end{equation*}
  with $\cX_\text{tc} = \bR$.
  The system filling the box $\mathtt{mf}$
  is a mechanical friction model
  as defined in~\cref{ex:mf}
  and
  it gives a relation
  \begin{equation*}
    \cO_\text{mf} : 1 \times 1 \nrightarrow \bT \cX_\text{mf}
  \end{equation*}
  with $\cX_\text{mf} = \bR \times \bR$.
  The system filling the box $\mathtt{ht}$
  is a heat transfer model
  as defined in~\cref{ex:ht}
  and
  it gives a relation
  \begin{equation*}
    \cO_\text{ht} : 1 \times 1 \nrightarrow \bT \cX_\text{ht}
  \end{equation*}
  with $\cX_\text{ht} = \bR \times \bR$.
  Let the Cartesian product of the three relations
  be denoted by
  \begin{equation*}
    \cS_\text{tc} \times \cO_\text{mf} \times \cO_\text{ht}
    \colon
    \mathrm{T} \cX_\text{tc} \times 1 \nrightarrow
    \bT \cX_\text{tc} \times \bT \cX_\text{mf} \times \bT \cX_\text{ht}
    \,.
  \end{equation*}
  Further,
  let the semantics of the interconnection pattern
  be given by the relation
  \begin{equation*}
    \cP_P \colon
    \bT \cX_\text{tc} \times \bT \cX_\text{mf} \times \bT \cX_\text{ht}
    \nrightarrow
    \cP_I
    \,,
  \end{equation*}
  where
  $\cP_I = \bT \cX_\mathtt{p} \times \bT \cX_\mathtt{s} = \bT \bR \times \bT \bR$
  is the bundle of port variables of the outer interface.
  Then,
  the semantics of the composite system is given by
  the relation
  \begin{equation*}
    \cC
    \: = \:
    \cP_P \circ
    \left( \cS_\text{tc} \times \cO_\text{mf} \times \cO_\text{ht} \right)
  \end{equation*}
  from $\mathrm{T} \cX_\text{tc} \times 1$
  to $\cP_I$.
  \label{ex:composite}
\end{example}

For an isolated composite system
the assigned relation
$\cC \colon \mathrm{T} \cX_{I_s} \times \mathrm{T} \cX_{I_e} \nrightarrow 1$
simply amounts to a subbundle of
$\mathrm{T} \cX_{I_s} \times \mathrm{T} \cX_{I_e}$,
which implicitly defines
a vector field
on $\cX_{I_s} \times \cX_{I_e}$.
Given an initial condition,
one can integrate this vector field
to determine the evolution of the system.
We simply say that
$(x, \, \dot{x}) \in \cC$,
where
$x = (x_s, \, x_e)$,
gives the dynamics of the isolated system.
The following definition generalizes this statement.

\begin{definition}
  Let $s \in \SystemSet$ be a system
  with interface $I$
  and
  let
  the relation
  $
  \cR \colon \mathrm{T} \cX_{I_s} \times \mathrm{T} \cX_{I_e}
  \nrightarrow
  \cP_I
  $
  define its semantics,
  according to
  the definitions mentioned in~\cref{tab:systems}.
  Further,
  let $(x_s, \, x_e) \in \cX_{I_s} \times \cX_{I_e}$
  denote the state variables associated to
  all (nested) storage and environment components.
  If the system is a reversible or an irreversible component
  or given by an interconnection of only such components
  this is vacuous.
  Moreover,
  let $(x, f, e) \in \cP_I$
  denote the port variables of its (outer) interface.
  If the system is isolated this is vacuous.
  Then,
  the \textbf{dynamics of system} $s$
  is implicitly determined by
  \begin{equation*}
    \bigl(
      (x_s, \, \dot{x}_s), \,
      (x_e, \, \dot{x}_e), \,
      (x, \, f, \, e)
    \bigr)
    \in \cR
    \,.
    \qedhere
  \end{equation*}
\end{definition}

\begin{example}
  The dynamics of
  the composite system described in~\cref{ex:composite}
  is given by
  \begin{equation*}
    \bigl(
      (s, \dot{s}), \,
      (\mathtt{p.x}, \, \mathtt{p.f}, \, \mathtt{p.e}), \,
      (\mathtt{s.x}, \, \mathtt{s.f}, \, \mathtt{s.e})
    \bigr)
    \in \cC
    \,.
  \end{equation*}
  This boils down to the following equations:
  \begin{align*}
    \dot{s}
    \: &= \:
    \frac{1}{\theta_1} \left(
      d \, \upsilon^2
      \, + \,
      \alpha \, (\theta_2 - \theta_1)
    \right)
    \\
    \mathtt{p.f}
    \: &= \:
    d \, \upsilon
    \\
    \mathtt{s.f}
    \: &= \:
    -\frac{1}{\theta_2} \,
    \alpha \, (\theta_1 - \theta_2)
  \end{align*}
  Here,
  $\upsilon = \mathtt{p.e}$
  is the velocity of the kinetic energy domain
  represented by the outer port $\mathtt{p}$, while
  $\theta_1 = \dd U(s)$
  is the absolute temperature of the thermal capacity
  with
  internal energy function
  $U : \bR \rightarrow \bR$
  and
  $\theta_2 = {\color{violet} \theta_0} + \mathtt{s.e}$
  is the absolute temperature of the thermal energy domain
  represented by the outer port $\mathtt{s}$.
  Further,
  $d \in \bR_{>0}$
  is the damping coefficient
  and
  $\alpha \in \bR_{>0}$
  is the heat transfer coefficient.
\end{example}

In the context of programming languages,
syntactically correct code
does not necessarily define
a semantically valid program.
Similarly, there are cases where
a compiler for the EPHS language
must raise an error.
For example,
no junction can have two or more connected ports
belonging to storage or environment components.
A more comprehensive discussion of the conditions
required for a composite system to be well-defined
is left for future work.

\section{Thermodynamic consistency}%
\label{sec:thermodynamic-consistency}

Akin to
metriplectic systems or the GENERIC,
the structure of EPHS ensures
thermodynamic consistency.
To simplify things a bit,
we focus on isolated systems.

\begin{proposition}
  The dynamics of
	an isolated EPHS respect
  the first law and the second law.
\end{proposition}

To prove this,
one could first show that
the semantics of any (isolated) composite system
can be represented by
some abstract system of evolution equations.
Analogous to metriplectic systems or the GENERIC,
one could then use
the structural properties of
this abstract system of equations
to demonstrate thermodynamic consistency.
The following proof,
however, does not require this intermediate step.

\begin{proof}
	As the semantics
	associated to interconnection patterns
	is functorial,
	it suffices to
	show thermodynamic consistency
	for composite systems,
	where all subsystems are primitive.
	In terms of semantics,
	environment components are
	a special case of storage components.
	Any finite number of storage components
	can be combined into a single component
	by taking the named sum of their interfaces
	and by summing their energy functions.
	Similarly,
	any finite number of reversible/irreversible components
	can be combined into a single such component
	by taking the named sum of their interfaces
	and by taking the Cartesian product of the relations that define them.
	It hence suffices to
	show thermodynamic consistency
	for a composite system with
	one storage component,
	one reversible component, and
	one irreversible component.
  \begin{enumerate}
		\item{
			First, we consider the storage component
			and focus specifically on
			the structure of its exergy function $H$,
			see~\cref{def:exergy}.
			The following two facts follow immediately:
			\begin{enumerate}
				\item{%
					The energy $E$ must be conserved
					if the exergy $H$
					and all quantities
					also present in the reference environment
					are conserved.
				}
				\item{%
					Entropy must be growing at a non-negative rate
					if the exergy $H$ is decreasing at a non-negative rate,
					while the energy $E$
					and all quantities
					also present in the reference environment,
					except for entropy,
					are conserved.
				}
			\end{enumerate}
		}
    \item{%
			Regarding the interconnection
			of the components
      based on~\cref{def:semantics_pattern},
			we know that
			at each junction
			the flow variables balance,
			while the effort variables are equal.
			We can conclude the following:
			\begin{enumerate}
				\item{%
					The interconnection is power-preserving,
					i.e.~it conserves exergy.
				}
				\item{%
					The interconnection conserves
					the quantities associated with the ports
					because each flow variable represents
					the rate of change of the corresponding quantity,
					and the flow variables balance.
				}
				\item{%
					The interconnection conserves energy.
					This follows from the previous two points
					and point 1.~(a)
				}
			\end{enumerate}
    }
    \item{%
			Next,
			we consider the reversible component
			defined by a Dirac structure $\cD$,
			see~\cref{def:dirac}.
			Since $\cD$ is a power-preserving relation,
			it conserves the exergy $H$.
			One of the conditions in~\cref{def:reversible}
			guarantees that $\cD$ conserves
			all quantities
			also present in the reference environment,
			in particular entropy.
			Considering also point 1.~(a),
			it follows that $\cD$
			must also conserve the energy $E$.
    }
    \item{%
			Finally,
			we consider the irreversible component
			defined by an Onsager structure $\cO$,
			see~\cref{def:onsager_structure}.
			As a consequence of the non-negative definiteness property,
			$\cO$ destroys exergy at a non-negative rate.
			One of the conditions in~\cref{def:irreversible}
			guarantees that $\cO$ conserves
			the energy $E$.
			Another condition guarantees that $\cO$ conserves
			all quantities
			also present in the reference environment,
			except for entropy.
			Considering also point 1.~(b),
			it follows that $\cO$
			must produce entropy at a non-negative rate.
    }
  \end{enumerate}
	The composite system must respect
	the first law, since
	the interconnection,
	the reversible component, and
	the irreversible component
	all conserve energy.
	\\
  The composite system must respect
	the second law, since
	the interconnection and
	the reversible component
	conserve entropy,
	while the irreversible component
	produces entropy at a non-negative rate.
\end{proof}

The concept of a \textit{storage function}
was introduced in
Willems' theory of dissipative systems~%
\cite{1972Willems}.
In this framework,
EPHS without irreversible components
are classified as \textit{lossless},
whereas EPHS that include irreversible components
are considered \textit{dissipative} or \textit{passive}.
The additional conditions
defining reversible and irreversible components
in EPHS
extend the notion of dissipativity,
ensuring full thermodynamic consistency.

\section{Example: shunt motor}%
\label{sec:example}

This section
compares
a bond graph
and
an EPHS model
of a direct current (DC) shunt motor.
A DC motor
has a fixed outer part,
called stator,
and a rotating inner part,
called rotor.
Both parts contain a coil of wire
through which current flows.
In the case of a shunt motor,
these two inductors are connected
in parallel
to a source of electric energy.
The working principle of such motors
is based on
a coupling of
the electromagnetic and the kinetic energy domains
via
the Lorentz force
and
the commutation of the rotor coil
after every half turn.
However,
both aspects are not directly resolved
by the considered models.
Instead, the coupling
is considered to be lumped into a gyrator element.

\begin{figure}[ht]
  \centering
  \includegraphics[width=0.9\textwidth]{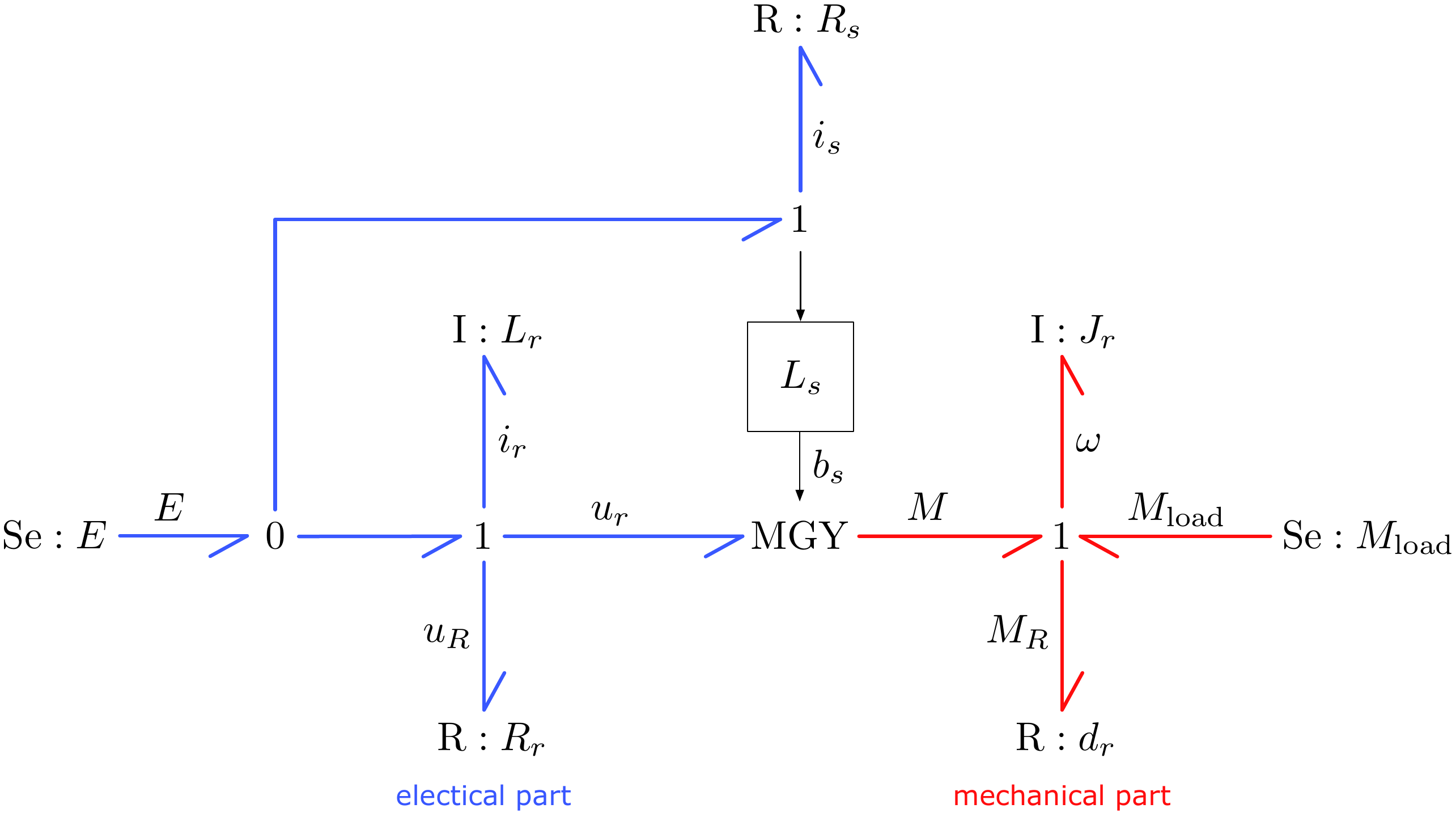}
  \caption{%
    Bond graph model of a DC shunt motor.
  }%
  \label{fig:motor_bond_graph}
\end{figure}

We start with the bond graph model
that is
taken from~\cite{2010Borutzky}
and
reproduced,
with slightly different naming of variables,
in~\cref{fig:motor_bond_graph}.
The half arrows are called bonds
and
each bond has two associated variables,
called flow and effort.
The name of the flow variable is written on
the side of the arrow head,
while the name of the effort variable is written on
the opposite side.
However,
not all variables are explicitly named in the example.
At any $0$-junction,
the effort variables of all connected bonds are equal,
while the flow variables balance,
with their sign given by the bond direction.
Dually,
at $1$-junctions,
flows are equal and efforts balance.
%
Specifically,
at the $0$-junction on the left,
the current from the $\mathrm{Se}$-element (effort source)
with effort/voltage $E$
splits into
the current $i_s$ through the stator coil
and
the current $i_r$ through the rotor coil.
The inductance of the rotor coil
is represented by the $\mathrm{I}$-element (generalized inductor)
with inductance $L_r$,
while
an $\mathrm{I}$-element representing
the stator coil inductance is omitted.
The (steady-state) stator coil current $i_s$ is
instead entirely determined by
the $\mathrm{R}$-element (generalized resistor)
at the top representing
the coil's electrical resistance $R_s$.
The magnetic flux of the stator coil $b_s$
is then obtained as the signal
that results from multiplying
the signal $i_s$,
which is the shared effort variable at the upper $1$-junction,
with the constant gain $L_s$.
The bond graph hence combines energy-based modeling
with signal-based modeling based on block diagrams.
At the $1$-junction on the left,
the source voltage $E$
is equal to
the sum of
the voltage over the rotor coil inductance $L_r$,
the voltage $u_R$ over the rotor coil resistance $R_r$,
and the induced (back EMF) voltage $u_r$.
The latter is determined by
the $\mathrm{MGY}$-element (modulated gyrator),
which describes the coupling of
the electrical and the mechanical parts.
The gyrator ratio is the signal $b_s$
coming from the gain block.
At the gyrator,
the effort variable of one port is equal to
the product of
the gyrator ratio
and
the flow variable of the other port.
On the mechanical side,
the load torque $M_\text{load} \leq 0$
is modeled by the effort source on the right,
while
the rotor's angular mass $J_r$
is represented by an $\mathrm{I}$-element.
At the $1$-junction
the sum of
the generated torque $M$
and
the load torque $M_\text{load}$
is equal to
the sum of
the torque accelerating the angular mass
and
the friction torque $M_R$.
The latter is
the product of
the friction coefficient $d_r$
and
the angular velocity $\omega$.
Putting all this together,
the bond graph gives
the system of ordinary differential equations
\begin{align*}
  L_r \, \dv{i_r}{t}
  \: &= \:
  E
  - u_r
  - u_R
  \\
  J_r \, \dv{\omega}{t}
  \: &= \:
  M
  - M_R
  + M_\text{load}
  \,,
\end{align*}
where
\begin{align*}
  i_s
  \: &= \:
  \frac{1}{R_s} \, E
  \\
  b_s
  \: &= \:
  L_s \, i_s
  \\
  u_r
  \: &= \:
  b_s \, \omega
  \\
  M
  \: &= \:
  b_s \, i_r
  \\
  M_R
  \: &= \:
  d_r \, \omega
  \\
  u_R
  \: &= \:
  R_r \, i_r
  \,.
\end{align*}

As already shown in~\cref{ex:motor_pattern},
the EPHS model
is specified as
an interconnection of two subsystems
named $\mathtt{stator}$ and $\mathtt{rotor}$.
The two systems
have a common
electric energy domain
represented by the junction
with connected ports named $\mathtt{q}$.
The kinetic energy domain of the rotor
is represented by the junction
with connected ports named $\mathtt{p}$.
Both energy domains are exposed
such that the motor model
can be integrated into another system.
The conversion between electric and kinetic energy
in the rotor
depends on the state of
the magnetic energy domain of the stator,
which is represented by the junction
with connected state ports named $\mathtt{b_s}$.

\begin{figure}[ht]
  \centering
  \includegraphics[width=27em]{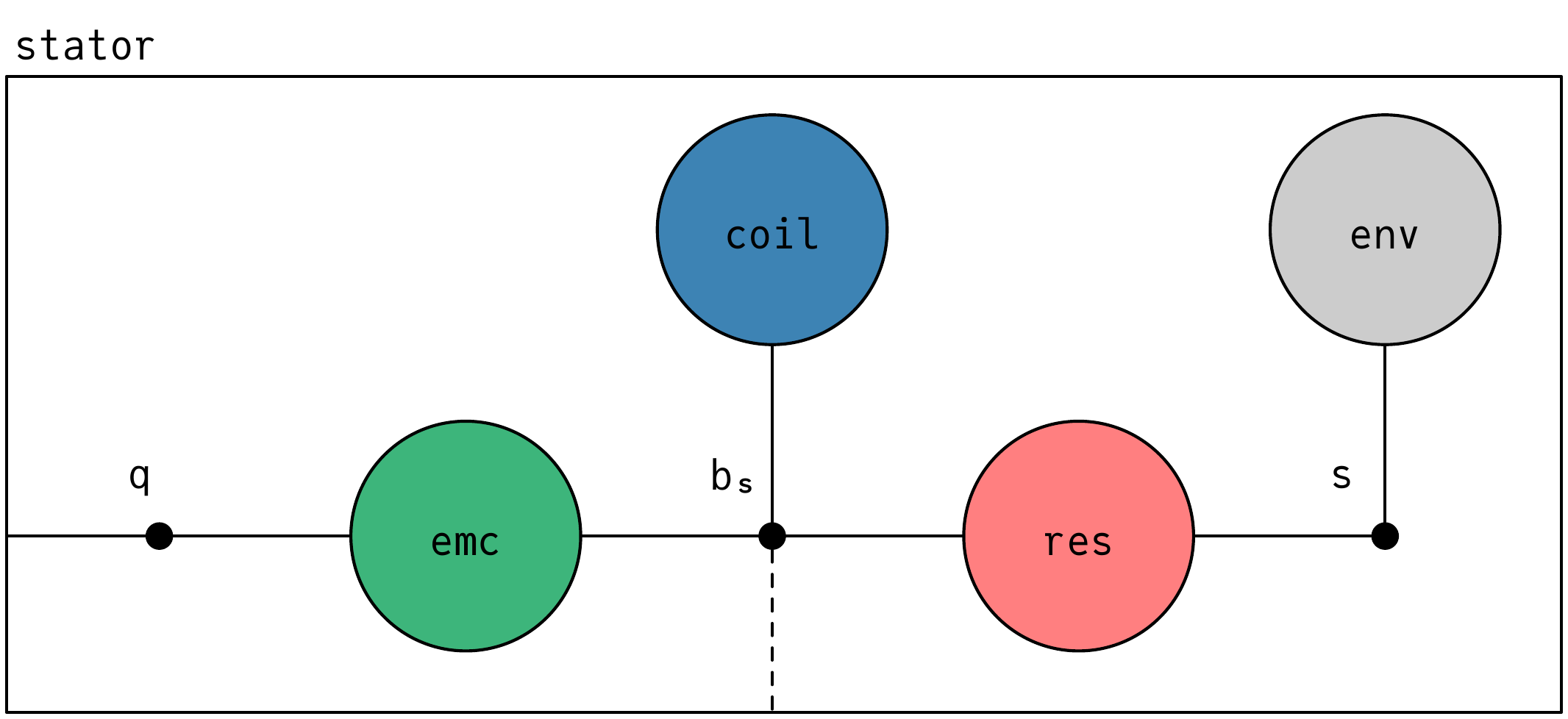}%
  \caption{%
    Interconnection pattern of the stator model.
  }%
  \label{fig:stator}
\end{figure}

The interconnection pattern of the stator model
is shown in~\cref{fig:stator}.
Box $\mathtt{emc}$
is filled by a reversible component
that represents
the coupling of
the electric energy domain associated with port $\mathtt{q}$
and the magnetic energy domain associated with port $\mathtt{b_s}$.
Box $\mathtt{coil}$
is filled by a storage component
that represents
the inductance of the stator coil.
Further,
box $\mathtt{res}$
is filled by an irreversible component
that represents
the resistance of the coil.
Finally,
box $\mathtt{env}$
is filled by an environment component
whose thermal energy domain associated with port $\mathtt{s}$
absorbs
the dissipated electromagnetic energy.

The reversible component
$(I_\text{emc}, \, \cD_\text{emc})$
is defined by its interface
\\ 
$I_\text{emc} = ( \{ \mathtt{q}, \, \mathtt{b_s} \}, \, \tau_\text{emc})$
with
\begin{alignat*}{2}
  &\tau_\text{emc}(\mathtt{q})
  \: &&= \:
  ((\bR, \, \mathtt{charge}), \, \mathsf{p})
  \\
  &\tau_\text{emc}(\mathtt{b_s})
  \: &&= \:
  ((\bR, \, \mathtt{flux\_linkage}), \, \mathsf{p})
\end{alignat*}
and its Dirac structure
$\cD_\text{emc} \subset \cP_{I_\text{emc}}$
given by
\begin{equation*}
  \begin{bmatrix}
      \mathtt{q.f} \\
      \mathtt{b_s.f}
  \end{bmatrix}
  \: = \:
  \begin{bmatrix}
      0 & 1 \\
      -1 & 0
  \end{bmatrix}
  \,
  \begin{bmatrix}
      \mathtt{q.e} \\
      \mathtt{b_s.e}
  \end{bmatrix}
  \,.
\end{equation*}

The storage component
$(I_\text{coil}, \, E_\text{coil})$
is defined by its interface
$
I_\text{coil} =
(\{ \mathtt{b_s} \}, \, \tau_\text{coil})
$
with
$
\tau_\text{coil}(\mathtt{b_s}) =
((\bR, \, \mathtt{flux\_linkage}), \, \mathsf{p})
$
and its energy function
$E_\text{coil} : \cX_{I_\text{coil}} \rightarrow \bR$
with
$\cX_{I_\mathtt{coil}} = \bR \ni \mathtt{b_s.x} = b_s$.
The function is given by
\begin{equation*}
  E_\mathtt{coil}(b_s)
  \: = \:
  \frac{1}{2 \, L_s} \, b_s^2
  \,,
\end{equation*}
where
the parameter
$L_s$ is the inductance.
%
The effort variable
$\mathtt{b_s.e} = \frac{b_s}{L_s}$
is the current through the coil.

The irreversible component
$(I_\text{res}, \, \cO_\text{res})$
is defined by
its interface
$
I_\text{res} =
(\{ \mathtt{b_s}, \, \mathtt{s} \}, \, \tau_\text{res})
$
with
\begin{alignat*}{2}
  &\tau_\text{res}(\mathtt{b_s})
  \: &&= \:
  ((\bR, \, \mathtt{flux\_linkage}), \, \mathsf{p})
  \\
  &\tau_\text{res}(\mathtt{s})
  \: &&= \:
  ((\bR, \, \mathtt{entropy}), \, \mathsf{p})
\end{alignat*}
and its Onsager structure
$\cO_\text{res} \subset \cP_\text{res}$
given by
\begin{equation*}
  \begin{bmatrix}
    \mathtt{b_s.f} \\
    \mathtt{s.f}
  \end{bmatrix}
  \: = \:
  \frac{1}{\color{violet} \theta_0} \,
  R_s \,
  \begin{bmatrix}
    \theta & -i \\
    -i & \frac{i^2}{\theta}
  \end{bmatrix}
  \,
  \begin{bmatrix}
      \mathtt{b_s.e} \\
      \mathtt{s.e}
  \end{bmatrix}
  \: = \:
  \begin{bmatrix}
      R_s \, i \\
      -\frac{1}{\theta} \, R_s \, i^2
  \end{bmatrix}
  \,,
\end{equation*}
where
the parameter $R_s$ is
the resistance of the coil,
$i = \mathtt{b_s.e}$ is
the current through the coil, and
$\theta = \textcolor{violet}{\theta_0 +} \mathtt{s.e}$ is
the absolute temperature at which heat is dissipated
into the thermal energy domain.

The environment component
is defined by
the subinterface
$I_\text{env} = \{ \mathtt{s} \} \subseteq I_R$.

\begin{figure}[ht]
  \centering
  \includegraphics[width=27em]{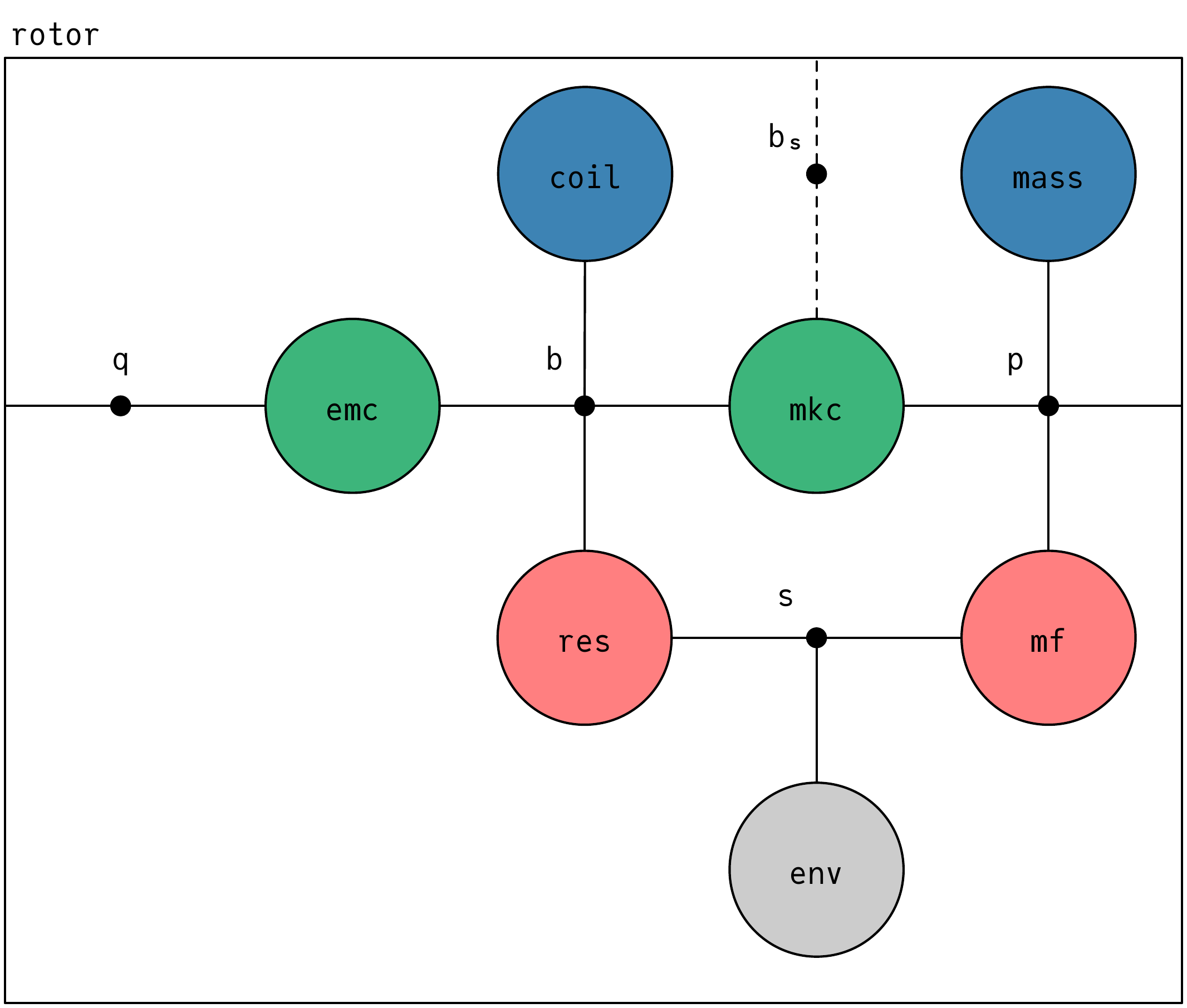}%
  \caption{%
    Interconnection pattern of the rotor model.
  }%
  \label{fig:rotor}
\end{figure}

The interconnection pattern of
the rotor model
is shown in~\cref{fig:rotor}.
The components
filling the boxes
$\mathtt{coil}$,
$\mathtt{emc}$,
$\mathtt{res}$ and
$\mathtt{env}$
are defined
as for the stator model,
except for
a different port name, namely
$\mathtt{b}$ instead of $\mathtt{b_s}$,
and
different parameters, namely
$L_r$ instead of $L_s$
and
$R_r$ instead of $R_s$.
Further,
box $\mathtt{mkc}$
is filled by a reversible component
that describes
the coupling of
the magnetic energy domain
associated with port $\mathtt{b}$
and the (rotational) kinetic energy domain
associated with port $\mathtt{p}$.
This coupling depends on
the magnetic flux $\mathtt{b_s.x}$
that is generated by the stator coil.
Box $\mathtt{mass}$
is filled by a storage component
that represents
the rotor's angular mass
and
box $\mathtt{mf}$
is filled by an irreversible component
that represents
its mechanical friction.

The reversible component
$(I_\text{mkc}, \cD_\text{mkc})$
is defined by
its interface
\\ 
$
I_\text{mkc} =
(\{ \mathtt{b}, \, \mathtt{p}, \, \mathtt{b_s} \}, \, \tau_\text{mkc})
$
with
\begin{alignat*}{2}
  &\tau_\text{mkc}(\mathtt{b})
  \: &&= \:
  ((\bR, \, \mathtt{flux\_linkage}), \, \mathsf{p})
  \\
  &\tau_\text{mkc}(\mathtt{p})
  \: &&= \:
  ((\bR, \, \mathtt{angular\_momentum}), \, \mathsf{p})
  \\
  &\tau_\text{mkc}(\mathtt{b_s})
  \: &&= \:
  ((\bR, \, \mathtt{flux\_linkage}), \, \mathsf{s})
\end{alignat*}
and its Dirac structure
$\cD_\mathtt{mkc} \subset \cP_{I_\text{mkc}}$
given by
\begin{equation*}
  \begin{bmatrix}
      \mathtt{b.f} \\
      \mathtt{p.f}
  \end{bmatrix}
  \: = \:
  \begin{bmatrix}
      0 & \mathtt{b_s.x} \\
      -\mathtt{b_s.x} & 0
  \end{bmatrix}
  \,
  \begin{bmatrix}
      \mathtt{b.e} \\
      \mathtt{p.e}
  \end{bmatrix}
  \,.
\end{equation*}

The storage component
$(I_\text{mass}, \, E_\text{mass})$
is defined by
its interface
$
I_\text{mass} =
(\{ \mathtt{p} \}, \, \tau_\text{mass})
$
with
$
\tau_\text{mass}(\mathtt{p}) =
((\bR, \, \mathtt{angular\_momentum}), \, \mathsf{p})
$
and
its energy function
\\ 
$E_\text{mass} : \cX_{I_\text{mass}} \rightarrow \bR$
with
$\cX_{I_\text{mass}} = \bR \ni \mathtt{p.x} = p$.
The function is given by
\begin{equation*}
  E_\text{mass}(p)
  \: = \:
  \frac{1}{2 \, J_r} \, p^2
  \,,
\end{equation*}
where
the parameter $J_r$ is the angular mass.

The irreversible component
$(I_\text{mf}, \cO_\text{mf})$
is defined by
its interface
$
I_\text{mf} =
(\{ \mathtt{p}, \, \mathtt{s} \}, \, \tau_\text{mf})
$
with
\begin{alignat*}{2}
  &\tau_\text{mf}(\mathtt{p})
  \: &&= \:
  ((\bR, \, \mathtt{angular\_momentum}), \, \mathsf{p})
  \\
  &\tau_\text{mf}(\mathtt{s})
  \: &&= \:
  ((\bR, \, \mathtt{entropy}), \, \mathsf{p})
\end{alignat*}
and
its Onsager structure
$\cO_\text{mf} \subset \cP_{I_\text{mf}}$
given by
\begin{equation*}
  \begin{bmatrix}
    \mathtt{p.f} \\
    \mathtt{s.f}
  \end{bmatrix}
  \: = \:
  \frac{1}{\color{violet} \theta_0} \,
  d_r \,
  \begin{bmatrix}
    \theta & -\upsilon \\
    -\upsilon & \frac{\upsilon^2}{\theta}
  \end{bmatrix}
  \,
  \begin{bmatrix}
    \mathtt{p.e} \\
    \mathtt{s.e}
  \end{bmatrix}
  \: = \:
  \begin{bmatrix}
    d_r \, \upsilon \\
    -\frac{d_r \, \upsilon^2}{\theta}
  \end{bmatrix}
  \,,
\end{equation*}
where the parameter $d_r$ is
the mechanical friction coefficient,
$\upsilon = \mathtt{p.e}$ is
the angular velocity, and
$\theta = \textcolor{violet}{\theta_0 +} \mathtt{s.e}$ is
the absolute temperature at which heat is dissipated
into the thermal energy domain.

We write
$(b_s, \, b, \, p, \, s) \in \bR^4$
for the state of the interconnected motor model,
where we identify
the entropy state variables $s$ of
both environment components.
In a future work,
the identification of
environment components
can be formalized by
composing with a suitable relation.
The EPHS model then yields the equations
\begin{align*}
  \dot{b}_s
  \: &= \:
  \mathtt{q.e} - R_s \, i_s
  \\
  \dot{b}
  \: &= \:
  \mathtt{q.e} - R_r \, i_r -  b_s \, \omega
  \\
  \dot{p}
  \: &= \:
  b_s \, i_r - d_r \, \omega + \mathtt{p.f}
  \\
  \dot{s}
  \: &= \:
  \frac{1}{\textcolor{violet}{\theta_0}}
  \left(
    R_s \, i_s^2 +
    R_r \, i_r^2 +
    d_r \, \omega^2
  \right)
  \\
  \mathtt{q.f}
  \: &= \:
  i_s + i_r
  \\
  \mathtt{p.e}
  \: &= \:
  \omega
  \,,
\end{align*}
where
\begin{align*}
  i_s
  \: &= \:
  \frac{1}{L_s} \, b_s
  \\
  i_r
  \: &= \:
  \frac{1}{L_r} \, b
  \\
  \omega
  \: &= \:
  \frac{1}{J_r} \, p
  \,,
\end{align*}
which are similar to
the equations for the bond graph model,
except that
(1)
the source of electrical energy
and
the mechanical load
are regarded as external dynamical systems,
rather than internal constants,
and
(2)
no quasi-static assumption is hence made for
the magnetic field of the stator.

\section{Discussion}%
\label{sec:discussion}

\subsection{Comparison with bond graphs}%

Bond graphs provide a graphical notation
for representing physical systems
in terms of subsystems that exchange energy.
However,
the notation does not naturally support
the hierarchical nesting of systems.
Relying exclusively on primitive subsystems
limits the applicability of bond graphs to complex systems,
where a hierarchical breakdown into simpler parts
would facilitate both model construction and reuse.

Although the notation is not inherently composable,
one can introduce conventions
that allow for subsystems
defined by open bond graphs.
For example,
The convention in~\cite{2021GawthropPanCrampin}
considers bond graphs
that share a specific element in common
as a modules.
These modules can be `glued' together
in a parent-level bond graph
via an additional $0$-junction.
In this approach,
the common element
in each individual bond graph
is replaced by an external port;
the element is instantiated once
in a parent-level bond graph,
where it is connected to a $0$-junction.
The modules are represented by additional elements,
and their ports are also connected to the $0$-junction.
The direction of these bonds
must match the direction
of the corresponding bond (or port)
at the module level.

While port-Hamiltonian systems can be seen as
a geometric formalization of
the dynamical systems derived from
\textit{generalized bond graphs},
the algebraic formalization of
the bond graph notation itself
is not straightforward
(see, e.g.,~\cite{1989BirkettRoe,2017Coya}).
Some inherent complexity in the notation
is evident also in practice,
where the direction
in which exchange of power is counted as positive,
and sometimes also the role of a variable as flow versus effort,
must be decided on a bond-by-bond basis.
These decisions must either be left to the user
or managed through heuristic algorithms or conventions.
This can impede
seamless hierarchical decomposition
and reuse of subsystems.

Interconnection patterns
provide a composable, graphical syntax
for the EPHS modeling language.
Compared to bond graphs,
interconnection patterns
have an outer box,
and they exclusively rely on $0$-junctions.
Based on this simplification,
the mentioned issues are resolved:
\begin{itemize}
  \item{%
		Interconnection patterns have
		a clear mathematical structure,
		which supports
		their hierarchical (de-)composition.
	}
	\item{%
		Interconnection patterns are
		relatively easy to interpret,
		owing to their simple structure,
		the direct correspondence between
		junctions and energy domains,
		and the clear role of a power variables
		as either flow or effort.
	}
	\item{%
		The direction
		in which exchange of power is counted positive
		is predetermined
		in agreement with
		the convention used in thermodynamics,
		i.e.~energy supplied to a system has a positive sign.
	}
\end{itemize}

In contrast to bond graphs,
EPHS models are guaranteed
to be thermodynamically consistent.

Bond graphs may incorporate
elements from block diagrams.
In addition to bonds,
which represent energy exchange,
there may be edges that transmit signals
according to their orientation.
Elements may have signal ports
to receive signals,
which modulate their behavior,
and there may be `blocks'
that perform arbitrary signal processing.
In contrast,
the formalism presented here does not cover
arbitrary signal processing;
system behavior can be modulated
only by state variables of other systems.
To express this,
every port has a state variable,
which is shared with all other ports
connected to the same junction.
We distinguish between
state ports,
which share only state,
and power ports,
which additionally model energy exchange.

\subsection{Comparison with port-Hamiltonian systems}%

Port-Hamiltonian Systems (PHS)
provide a mathematical framework
for open dynamical systems
with guaranteed passivity.
While bond graphs are frequently used
to depict port-Hamiltonian models,
establishing a rigorous link
between bond graph syntax and port-Hamiltonian semantics
remains challenging.
On the one hand,
bond graphs offer
a network representation of individual systems,
but the representation is not inherently composable.
On the other hand,
PHS provide
a composable representation of individual systems,
albeit without an explicit internal network structure.
In particular,
the PHS framework provides no built-in support for
referencing subsystems and their ports.
Previous work has focused on
algorithms for converting
certain classes of bond graphs to PHS
(see, e.g.,~\cite{2020PfeiferCaspartHampelMullerKrebsHohmann}).

A well-known issue concerns
the relationship between port-Hamiltonian structure,
which implies passivity,
and its physical interpretation.
Owing to its roots in mechanics,
the `Hamiltonian' storage function
is typically interpreted as the system's energy.
However,
the differential of the `Hamiltonian'
generally also drives dissipative dynamics,
resulting in an apparent loss of `energy'.
This raises questions,
especially outside the realm of electromechanical systems,
where thermodynamic considerations are often overlooked.
In response,
some members of the port-Hamiltonian community
have proposed alternative port-based frameworks
that explicitly encode thermodynamic structure
(see, e.g.,~\cite{2005EberardMaschkeSchaft,2022RamirezGorrec,2019SchaftMaschke}).

EPHS provide a mathematical framework
for open dynamical systems
with guaranteed thermodynamic consistency.
In contrast to PHS,
this framework is directly based on
a composable, graphical syntax
used to specify how
more complex systems are formed by
interconnecting simpler systems.
An explicit understanding of
the translation between graphical syntax and relational semantics
guarantees that syntax and semantics compose in compatible ways.
This enables
a modular and hierarchical modeling approach,
which inherently supports
referencing nested subsystems and their ports.
In a sense,
this approach unites
the strengths of both bond graphs and PHS.

The dynamical systems defined by EPHS models
essentially possess port-Hamiltonian structure.
Beyond passivity,
additional structural properties
imposed on the primitive systems
guarantee that all EPHS
share a uniform and consistent physical interpretation,
aligned with nonequilibrium thermodynamics.
Since EPHS compose simply by sharing energy domains,
users typically do not need to explicitly define external ports.
As a result,
users can typically integrate existing models
into more complex systems without any prior adaptation.

\subsection{Comparison with metriplectic and GENERIC systems}%

Metriplectic systems and the GENERIC
provide a framework for (usually isolated) dynamical systems
with guaranteed thermodynamic consistency.
The dynamical systems
defined by (isolated) EPHS models
are basically metriplectic systems or instances of the GENERIC,
since the built-in thermodynamic consistency
follows from essentially the same ideas.
Regarding the reversible part,
EPHS appear to be more general,
since they use Dirac structures,
rather than the less general Poisson structures
used in the definition of metriplectic systems or the GENERIC.
This allows EPHS to easily model kinematic constraints,
which is important e.g.~for multibody systems
(see~\cite{2024LohmayerCapobiancoLeyendecker}).
However,
it should be no problem to generalize
the metriplectic or GENERIC framework accordingly.
Regarding the irreversible part,
EPHS appears to be less general,
since it remains unclear if
the single-generator formulation based on exergy
is useful beyond local thermodynamic equilibrium.

In contrast to
the metriplectic or GENERIC framework,
EPHS enable a modular and hierarchical approach
to thermodynamic modeling,
based on a graphical syntax
for building more complex models from simpler parts.

\subsection{Future work}%

Future work shall
address the details omitted in this work.
A more technical paper will define
interconnection patterns
and their functorial semantics
using a categorical framework
based on directories
and the theory of operads/multicategories.

As presented here,
the definition of (composite) EPHS
exhibits a recursive structure.
It would be interesting to study
the `closure' of all EPHS models
by deriving representations of EPHS
as structured systems of differential-algebraic equations.
The analysis could benefit
further developments
and further clarify the relationship with
the metriplectic or GENERIC framework
as well as
port-Hamiltonian systems.

As an important groundwork for
further practical developments,
we have implemented the framework
within the Julia programming language,
see~\cite{2025Lohmayer}.
The reader is invited to
review its application in several examples.
Based on this,
future work can explore topics such as
advanced applications
(see, e.g.,~\cite{2024LohmayerCapobiancoLeyendecker,2024LohmayerKrausLeyendecker}),
structure-preserving numerical methods,
semi-automated control design,
and the integration with scientific machine learning methods.

\section*{Author contribution statement}%

\textbf{Markus Lohmayer}: Conceptualization, Investigation, Writing -- Original Draft, Writing -- Review \& Editing, Visualization;
\textbf{Owen Lynch}: Investigation, Writing -- Original Draft;
\textbf{Sigrid Leyendecker}: Supervision, Review \& Editing, Funding

\section*{Acknowledgments}%

We thank
David Spivak for help
when working out the concept of namespaces.
The second author acknowledges support from
the DARPA ASKEM program under grant HR00112220038.


\bibliographystyle{link-elsarticle-num}
\bibliography{literature}
\addcontentsline{toc}{section}{References}

\end{document}